\documentclass[a4paper, onecolumn, 11pt, accepted=2021-07-20]{quantumarticle}
\pdfoutput=1
\usepackage{graphicx}
\usepackage{hyperref}
\usepackage{physics}
\usepackage{color}
\usepackage{enumerate}
\usepackage{bm}
\usepackage[normalem]{ulem}

\usepackage[caption=false]{subfig}

\usepackage[numbers, sort&compress]{natbib}
\usepackage{graphicx}
\usepackage{braket}
\usepackage{amsmath}

\usepackage{amsthm}
\usepackage{amsfonts, amssymb}
\usepackage{bbm}
\usepackage{hyperref}
\newtheorem{prop}{Proposition}

\usepackage{mathtools}
\usepackage{mathrsfs}

\DeclarePairedDelimiterX{\infdivx}[2]{(}{)}{%
	#1\;\delimsize\|\;#2%
}

\begin{document}

	\title{On the distribution of the mean energy in the unitary orbit of quantum states}
	\author{Raffaele Salvia} 
	\email[Corresponding author: ]{raffaele.salvia@sns.it}
	\affiliation{Scuola Normale Superiore and University of Pisa, I-56127 Pisa, Italy}
	
	\author{Vittorio Giovannetti}
	\affiliation{NEST, Scuola Normale Superiore and Istituto Nanoscienze-CNR, I-56127 Pisa, Italy}
	
	\begin{abstract}
Given a closed quantum system, the states that can be reached with a cyclic process are those with the same spectrum as the initial state. Here we prove that, under a very general assumption on the Hamiltonian, the distribution of the  mean  extractable work is very close to a gaussian with respect to the Haar measure.
	We derive bounds for both the moments of the distribution of the mean energy of the state and for its characteristic function, showing that the discrepancy with the normal distribution is increasingly suppressed for large dimensions of the system Hilbert space. 
	\end{abstract}

	\maketitle
	
	\section{Introduction}

	Haar-uniform random unitary matrices are a resource required for various quantum algorithm \cite{Hayashi2005, Scott2008, BrandaoHorodecki2008}. 
	As an example, the \emph{randomised benchmark} protocol is a method to test the error rate of a quantum circuit, requiring it to perform a sequence of random operations \cite{Emerson2005}. Versions of the randomised benchmark are used by the companies IBM \cite{McKay2019} and Microsoft \cite{Helsen2019} to test the functionality of their experimental quantum computing hardware.
	Other applications of random unitaries include quantum cryptography \cite{DiVincenzo2002} and the simulation of many body physics \cite{Nahum2017, Jonay2018}.

	In this paper we charachterise the the distribution of the energetic cost of implementing a random unitary transformation. To be more specific, we study the distribution  of the expected value of the energy gained from the cycle which sends the quantum state $\rho$ to $U\rho U^\dagger$, where $U$ is a matrix extracted randomly with respect to the Haar measure of the unitrary group. 
	Our main finding is that, when the quantum system has many degrees of freedom (that is when the Hilbert space has dimension $d \gg 1$), then -provided a weak condition on the system Hamiltonian-  the distribution of the energy cost of a random unitary matrix is approximately Gaussian.

	Implementing a random unitary is not an easy task. Actually, only a small subset of quantum operations (called the \emph{Clifford gates} \cite{Gottesman1998}) are easy to implement in an actual circuit - the complexity of a quantum circuit is often measured with the number of non-Clifford gates it requires \cite{Veitch2014}.
	To overcome these difficulties, it has been theorised the possibility of circuits that simulate a random unitary up to a certain moment of the distribution \cite{Gross2007}. A circuit which is able to emulate a uniform distribution up to the $t$-th moment is called an \emph{unitary t-design}. If we want to realise a $t$-design with $t > 3$, it is still necessary to use non-Clifford gates \cite{Sawicki2017, Bannai2018}, but only a small amount of them \cite{Haferkamp2020}.

	The results of this paper are useful for characterising an ideal source of uniformly distributed unitary operations. Furthermore, they are valid also for its approximations which are used in actual quantum computing, the $t$-designs which match the uniform unitary distribution up to the $t$-th moment.
	Indeed we will show that, if the dimension of the Hilbert space is big enough, and if it holds a very general condition, the first moments of the distribution are very well approximated by the moments of a gaussian distribution with the same variance.
	Exploiting this result, we will then be able to estimate the maximum error that we make in replacing the charachteristic function of distribution of the energy with the characteristic function of a normal distribution. 
 
	We emphasize that our paper is concerned with the distribution of the mean estracted work \emph{between different processes}, which is different from the distribution of extracted work \emph{in a given process} (a much more studied subject in in classical and quantum theormodynamics \cite{Esposito2009, CampisiHanggiTakner}).
	The latter distribution is in many cases Gaussian for classical systems in the quasistatic limit~\cite{Speck2004, ubrt2007, Speck2011}, but not quantum system~\cite{Miller2019}: in particular, as shown in Ref.~\cite{Scandi2020}, for slowly driven quantum-thermodynamical processes the work distribution becomes non-gaussian whenever quantum coherences are generated during the protocol.
	However, we may think to a ``process'' which consists in selecting a random unitary matrix according to the Haar measure and then applying it to the system: in this case, under the conditions mentioned above the work distribution would be Gaussian: we think it worth to notice that, when the initial state of the system is a state of thermal equilibrium (so that Jarzynski's equality \cite{jarzynski1997} holds), this implies the validity of the classical fluctuation-dissipation relation \cite{Esposito2009}.

	\section{Introduction to the problem} \label{sec1} 
	Let $A$ be a $d$-dimensional quantum system  initialized in the state $\hat{\rho}$ and forced to evolve in 
	time by  external modulations of its Hamiltonian $\hat{H}$. 
	Following~\cite{Pusz1978,Lenard1978}, the average amount of work one can extract from the process
	can be computed as
	\begin{equation}
	W_U(\hat{\rho};\hat{H}) := E(\hat{\rho};\hat{H}) - E(\hat{U} \hat{\rho} \hat{U}^\dagger; \hat{H}) \;,
	\label{extracted_workWu}
	\end{equation}
	where 
	\begin{eqnarray} E(\hat{\rho}; \hat{H}):= \mbox{Tr}[\hat{\rho} \hat{H}]\label{MEANE}\;,\end{eqnarray}  is the mean energy of $\hat{\rho}$, and $\hat{U}$ is the element of the unitary group $\mathbf{U}(d)$ 
	associated to the applied driving.
The allowed values of $W_U(\hat{\rho};\hat{H})$ 
 are limited by the inequalities
\begin{eqnarray} \label{boundaries} 
{\cal A}(\hat{\rho};\hat{H})   \leq W_U(\hat{\rho};\hat{H}) \leq {\cal E}(\hat{\rho};\hat{H})  \;, 
\end{eqnarray} 
where ${\cal E}(\hat{\rho};\hat{H})$ and  ${\cal A}(\hat{\rho};\hat{H})$
 are respectively the  {\it ergotropy} and {\it anti-ergotropy} functionals of the model.
The first one was introduced in Ref.~\cite{Allahverdyan2004} and corresponds to maximum work one can 
extract from the system by optimizing  $W_U(\hat{\rho};\hat{H})$ with respect to the choices of the control unitary $\hat{U}$ for fixed
$\hat{\rho}$ and $\hat{H}$.  
A closed formula for the ergotropy is provided by the expression  
\begin{eqnarray}
 {\cal E}(\hat{\rho};\hat{H})  &:=&   \max_UW_U(\hat{\rho};\hat{H}) = W_{\hat{U}^{(\downarrow)} } (\hat{\rho}; \hat{H})  \nonumber \\
 &=&  E(\hat{\rho};\hat{H}) - E(\hat{\rho}^{(\downarrow)} ; \hat{H}) \;,
	\label{ergotropy}
	\end{eqnarray}
obtained by setting $\hat{U}$ equal to the optimal unitary  $\hat{U}^{(\downarrow)}$ for which 
 the final state of the cycle $\hat{U}\hat{\rho}\hat{U}^{\dagger}$  corresponds to the passive counterpart  $\hat{\rho}^{(\downarrow)}$ of $\hat{\rho}$~\cite{Pusz1978, Lenard1978}, i.e.  
\begin{equation} 
\hat{U}^{(\downarrow)}: =  \sum_{j=1}^d |\epsilon^{(\uparrow)}_j\rangle\langle \lambda_j^{(\downarrow)}  |\;,  
\quad 
\hat{\rho}^{(\downarrow)} :=  \sum_{j=1}^d \lambda_j^{(\downarrow)} |\epsilon^{(\uparrow)}_j\rangle\langle\epsilon^{(\uparrow)}_j|\;,\label{passive} 
\end{equation} 
with $|\lambda_j^{(\downarrow)}  \rangle$ the eigenvector of $\hat{\rho}$ associated with the eigenvalue $\lambda_j^{(\downarrow)}$
listed in non-increasing order (i.e. 
$\lambda_j^{(\downarrow)} \geq \lambda_{j+1}^{(\downarrow)}$), and $|\epsilon^{(\uparrow)}_j\rangle$ 
the eigenvector of $\hat{H}$ with eigenvalue $\epsilon^{(\uparrow)}_j$ listed instead in non-decreasing order 
(i.e. 
$\epsilon_j^{(\uparrow)} \leq \epsilon_{j+1}^{(\uparrow)}$). 
Similarly for the   anti-ergotropy we can write
\begin{eqnarray} {\cal A
	}(\hat{\rho};\hat{H})  &:=&   \min_UW_U(\hat{\rho};\hat{H}) =  W_{\hat{U}^{(\uparrow)} } (\hat{\rho};\hat{H})  \nonumber \\ 
	&=& E(\hat{\rho};\hat{H}) - E(\hat{\rho}^{(\uparrow)} ; \hat{H}) \;,
	\label{a-ergotropy}
	\end{eqnarray}
where now $\hat{\rho}^{(\uparrow)}$ is the anti-passive state of $\hat{\rho}$, 
obtained by the unitary $\hat{U}^{(\uparrow)}$ that   reverses the order in which the  populations of  $\hat{\rho}^{(\downarrow)}$ are listed, i.e. 
\begin{equation} 
\hat{U}^{(\uparrow)}: =  \sum_{j=1}^d |\epsilon^{(\uparrow)}_j\rangle\langle \lambda_j^{(\uparrow)}  |\;,  
\quad 
\hat{\rho}^{(\uparrow)} := \sum_{j=1}^d \lambda_{j}^{(\uparrow)} |\epsilon^{(\uparrow)}_j\rangle\langle\epsilon^{(\uparrow)}_j|\;,
\end{equation} 
with $\lambda_j^{(\uparrow)} := \lambda_{d-j+1}^{(\downarrow)}$. 

Saturating the upper bound~(\ref{boundaries}) is  an important optimization task which can be practically difficult to implement, as it  implicitly requires an exact knowledge of the full spectral decomposition of the system Hamiltonian. 
In this perspective, it is  interesting to understand how close  one  can get  from the boundary values~(\ref{boundaries}) by randomly selecting 
$\hat{U}$  for a model in which  both $\hat{\rho}$ and $\hat{H}$ are assigned. Clearly, 
as the dimensionality of the system increases, we do not expect such a naive approach to be particularly effective.
 Still, providing an exact characterization of the associated efficiency   is a well-posed statistical question which can be of some help in identifying which physical systems
are best suited as successful candidates for implementing quantum battery models~\cite{Campaioli2018, AlickiFannes2013, Hovhannisyan2013, Binder2015, JuliFarr2020}.  
In order to tackle this issue, in the present paper we 
study 
 the probability distribution  $P(E| \hat{\rho}; \hat{H})$ 
   of the mean output energy 
\begin{eqnarray} \label{RANDOME} 
E := E(\hat{U} \hat{\rho} \hat{U}^\dagger; \hat{H}) \;,
\end{eqnarray} 
 which originates by random sampling $\hat{U}$ on the unitary group $\mathbf{U}(d)$ via its natural measure (the Haar measure $d\mu(\hat{U})$)~\cite{NOTA00}. A numerical example of this distribution, with $d=7$ is showed in Fig.~\ref{fig:distribuzione}.
 From Eq.~(\ref{extracted_workWu}) it is clear that knowing $P(E| \hat{\rho}; \hat{H})$ we can then reconstruct the 
 probability distribution 
 $P_{\text{work}}(W|  \hat{\rho}; \hat{H})$
of the average  extracted work $W:=W_U(\hat{\rho};\hat{H})$ via a simple shift of the argument, i.e. 
	 \begin{equation}
	P_{\text{work}}(W|  \hat{\rho}; \hat{H}) : = P\Big(E=E(\hat{\rho};\hat{H}) - W\Big|  \hat{\rho}; \hat{H}\Big)  \;.
	\label{extracted_workWuprob}
	\end{equation}  
	
\begin{figure} 
	\centering
	\includegraphics[width=\columnwidth]{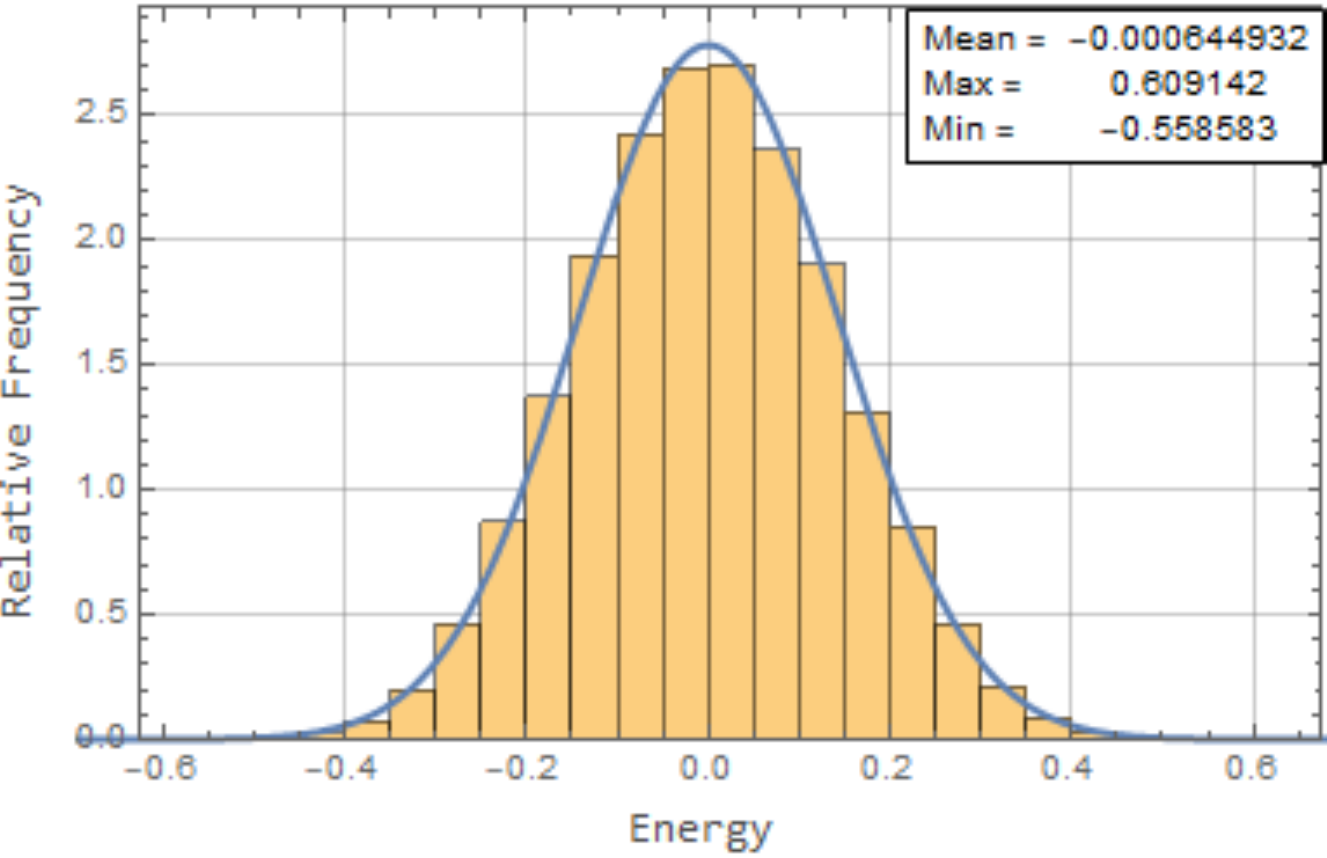}
	\caption{Numerical estimation of the distribution $P(E| \hat{\rho}; \hat{H})$ in a seven-levels ($d=7$) system in which $\hat{H}$ has eigenvalues  $\{ -1.6,-1.2,-0.6,0,0.4,1.3,1.7 \}$,  and $\hat\rho$ has eigenvalues
    $\{ 0.395,0.224,0.151,0.115,0.079,0.0020,0.013 \}$. 
  The histogram plots the empirical distribution of $E(\hat{U} \hat{\rho} \hat{U}^\dagger; \hat{H})$ for a sample of $10^5$ unitary matrices $U \in {\bf U}(7)$, which are distributed uniformly according to the Haar measure. The blue line is the probability density function of the Gaussian distribution $P_G^{(\mu,\Sigma^{(2)})}(E)$ with the variance $\Sigma^{(2)}$ given by~(\ref{varianza}).    In this example $\Sigma^{(2)} \simeq 0.02024$ and $\eta_{\hat{H}} \simeq 0.2317$.}
\label{fig:distribuzione}
\end{figure}

Our main finding is  that, under mild assumptions on the system Hamiltonian,  when the dimension of the Hilbert space $d$ is sufficiently large, the central moments of the probability $P(E| \hat{\rho}; \hat{H})$  
  can be approximated by those of a gaussian distribution $P_G^{(\mu,\Sigma^{(2)})}(E)$ having  mean value $\mu$ and variance $\Sigma^{(2)}$ equal to those of $P(E| \hat{\rho}; \hat{H})$, i.e. 
	\begin{equation}
	\label{intro_momenti_gaussiani}
	\Sigma^{(p)}  := \left\langle \left( E - \mu \right)^p	\right\rangle \simeq \Sigma_G^{(p)}\;, 
	\end{equation}
	where $\langle f(E) \rangle$ denotes the mean value  of the function $f(E)$  with respect to 
 $P(E| \hat{\rho}; \hat{H})$~\cite{NOTA00} and 
 \begin{eqnarray} \label{GAUSMOM} 
  \Sigma_G^{(p)} : = {\mathbf G}_{p} \times
  \begin{cases}
	1 &  \mbox{($p$ even)}\;, \\\\%
	0 & \mbox{($p$ odd)\;,}
	\end{cases}
 \end{eqnarray} 
 with  ${\mathbf G}_{p}$ the scaling factor 
 \begin{eqnarray} \label{SCALING} 
 {\mathbf G}_{p} : = (p-1)!! \; \left(\Sigma^{(2)}\right)^{p/2}\;.
 \end{eqnarray} 
In order to settle the approximation~(\ref{intro_momenti_gaussiani}) into firm quantitative ground  
 we prove that for $d$ sufficiently larger than $p$ 
  the discrepancies  between the l.h.s and the r.h.s of such equation can be bounded as
 \begin{eqnarray}
	{\left\lvert \Sigma^{(p)} -  \Sigma_G^{(p)} \right\rvert}
	\leq   {\mathbf G}_{p}  \;  f_{\hat{H}}(d,p)   \;, 
	\label{sezII_bbound_erroresusigma1}
	\end{eqnarray}
where $f_{\hat{H}}(d,p)$ is a positive function which 
in the large $d$ limit scales as $\mathcal{O}(1/d + \eta_{\Delta\hat{H}})$ for $p$ even, and  $\mathcal{O}(\sqrt{\eta_{\Delta\hat{H}}})$ for $p$ odd,
with  $\eta_{\Delta\hat{H}}$ being a functional (see Eq.~(\ref{FONDtermH}) below)
that  for typical choices of  the system Hamiltonian is very much depressed. 
Using the theory of generating functions, we also show a stronger result which, aside from bounding the error on each individual moment of the distribution
 $P(E| \hat{\rho}; \hat{H})$,  directly links  its characteristic function  with the one of the gaussian function $P_G^{(\mu,\Sigma^{(2)})}(E)$.

	The rest of the manuscript is organised as follows. 
	Section~\ref{sec:moments_1} introduces the concepts and some known results we will use in the proof. 
	In particular, Sec.~\ref{sec:momenti_defint} sets up the problem of calculating the moments of the distribution, which can be expressed as the Haar integral~(\ref{integrale_momenti}). The evaluation of this integral can be set up using the theory of Weingarten calculus, that we introduce in section \ref{sec:weingarten}. The practical computation of the integrals involved in the proof will require other combinatorical concepts, which will be introduced at the points where they are needed.
		In Sec.~\ref{sec:normaH} we derive some useful inequalities and discuss the assumptions on the system  that are needed to enforce the Gaussian approximation.
	Section~\ref{sec:outlook_statipuri} presents a proof of Eq.~(\ref{sezII_bbound_erroresusigma1}) that applies for the special case in which 
	the input state of the system $\hat{\rho}$ is pure, and in 
	Sec.~\ref{sec:bound_momenti} we address instead the case of arbitrary input states.	
	Finally, in Sec.~\ref{sec:moments_mgf} we  bound instead  the distance between the characteristic function $\langle e^{-itE} \rangle$ of the distribution $P(E| \hat{\rho}; \hat{H})$, and the characteristic function of the Gaussian distribution $P_G^{(\mu,\Sigma^{(2)})}(E)$. 
Conclusion and outlook are finally given in Sec.~\ref{Sec:conc}.

\section{Preliminary considerations}
\label{sec:moments_1}
This section is dedicated  to clarify some useful mathematical properties of the model.

\subsection{Basic properties of $P(E| \hat{\rho}; \hat{H})$} 

From Eqs.~(\ref{MEANE}) and ~(\ref{boundaries}) it follows that for fixed  $\hat{\rho}$ and $\hat{H}$ the range of the random variable $E$ is limited  by the inequalities  
\begin{equation} \mbox{Tr}[\hat{\rho}^{(\downarrow)}  \hat{H}]=
\sum_{j=1}^d  \lambda_j^{(\downarrow)} \epsilon^{(\uparrow)}_j \leq E \leq \sum_{j=1}^d  \lambda_j^{(\uparrow)} \epsilon^{(\uparrow)}_j = \mbox{Tr}[\hat{\rho}^{(\uparrow)}  \hat{H}] \;.
\label{sorting_inequality}
\end{equation}
In the special cases where either the input state is completely mixed (i.e.  $\hat\rho = {\hat{\openone}}/d$), or the Hamiltonian is proportional to the identity 
 (i.e.  $\hat{H} = E_0{\hat{\openone}}$), 
the upper and lower bounds of (\ref{sorting_inequality}) coincide
 forcing $E$ to assume the constant value $\mbox{Tr}[\hat H]/d$, i.e. imposing 
 the distribution $P(E| {\hat{\rho}}; \hat{H})$  to become a Dirac delta   \begin{eqnarray} \label{DELTA} 
P(E| {\hat{\openone}}/d; \hat{H}) = P(E| \hat{\rho};  E_0{\hat{\openone}}) =   \delta(E - \tfrac{\mbox{\small Tr}[\hat H]}{d})\;.
\end{eqnarray}

From Eq.~(\ref{RANDOME}) it is also clear that for arbitrary choices of  $\hat{\rho}$  any rigid shift of the Hamiltonian spectrum results in a translation of the distribution 
$P(E| \hat{\rho}; \hat{H})$, i.e. 
\begin{eqnarray} \label{SHIFT} 
P(E| \hat{\rho}; \hat{H}) = P\left(E-E_0\Big| \hat{\rho} ;\hat{H}-{E_0}{\hat{\openone}}\right)  \;.
\end{eqnarray} 
 More generally noticing that for all choices of $\Delta_0,E_0\in \mathbb{R}$ one has 
\begin{eqnarray} 
E(\hat{U} \hat{\rho} \hat{U}^\dagger; \hat{H}) &=& \mbox{Tr}\left[ \hat{U} \left(\hat{\rho}-\Delta_0 \tfrac{{\hat{\openone}}}{d}\right)\hat{U}^\dagger  \left(\hat{H}-E_0{\hat{\openone}}\right)\right]  \nonumber \\
&&+\Delta_0 \left( \tfrac{\mbox{Tr}[\hat{H}]}{d} -E_0\right) +E_0\;,
\end{eqnarray} 
 we can draw the following formal identity
\begin{eqnarray} \label{SHIFTnew} 
&&P(E| \hat{\rho}; \hat{H})  \\ \nonumber 
&&= P\left(E-\Delta_0 \left( \tfrac{\mbox{Tr}[\hat{H}]}{d} -E_0\right) -E_0\Big| \hat{\rho}-\Delta_0 \frac{{\hat{\openone}}}{d}; \hat{H} - E_0{\hat{\openone}} \right)\;,
\end{eqnarray} 
where, generalizing the definition  of $P(E| \hat{\rho}; \hat{H})$, given  $\hat{A}$ and $\hat{B}$ generic operators,
we use the symbol $P(E'| \hat{A}; \hat{B})$  to represent the distribution of the variable $E':=\mbox{Tr}[ \hat{U} \hat{A} \hat{U}^\dag \hat{B} ]$ induced by the Haar measure $d\mu(\hat{U})$. 

We observe next  that 
 any two input states 
 $\hat{\rho}'$ and $\hat{\rho}$ which have the same spectrum will have the same energy probability distribution, i.e. 
 $P(E|\hat{\rho}'; \hat{H})=P(E| \hat{\rho}; \hat{H})$. 
  Indeed 
 under this condition we can always express $\hat{\rho}'$ as $\hat{V} \hat{\rho} \hat{V}^\dag$ with $\hat{V}\in \mathbf{U}(d)$, so that 
 for all functions $f(E)$ one has
 \begin{eqnarray} 
\int d E f(E) P(E|\hat{\rho}'; \hat{H}) = \int d\mu(\hat{U})  f(E(U \hat{V} \hat{\rho} \hat{V}^\dag U^\dagger; \hat{H})) \nonumber \\
 =
 \int d\mu(\hat{U}')  f(E(\hat{U}'\hat{\rho} \hat{U}^{\prime \dagger}; \hat{H})) =\int d E f(E) P(E| \hat{\rho}; \hat{H}) \;, \nonumber 
 \end{eqnarray} 
 where we used the fact that $\hat{U}'=\hat{U}\hat{V} \in \mathbf{U}(d)$, and the invariance property $ d\mu(\hat{U}) = d\mu(\hat{U}\hat{V}) $ of the Haar measure. Similarly due to
 the cyclicity of the trace appearing in Eq.~(\ref{MEANE}), we can conclude that 
$P(E| \hat{\rho}; \hat{H})$ is also invariant under unitary rotations of the system Hamiltonian, leading to the identity  \begin{equation}
 P(E|V\rho V^\dag; W HW^\dag)=P(E| \hat{\rho}; \hat{H})\;, \quad   \forall V,W \in \mathbf{U}(d)\;,
\end{equation} 
which ultimately implies that $P(E| \hat{\rho}; \hat{H})$ can only depend upon the spectra $\{ \lambda_j^{(\downarrow)} \}_{i}$ and 
$\{ \epsilon_j^{(\uparrow)} \}_{i}$
of $\hat{\rho}$ and $\hat{H}$ but not on the specific choices of  their associated 
eigenstates $\{ |\lambda_j^{(\downarrow)} \rangle\}_{i}$ and $\{ |\epsilon_j^{(\uparrow)} \rangle \}_{i}$ nor on the relative overlap between them.

\subsection{Central moments} 
\label{sec:momenti_defint}

In evaluating the moments of the distribution $P(E;\hat{\rho},\hat{H})$ we have to consider the quantities   
\begin{eqnarray}\label{integrale_momenti}
\langle E^p\rangle &=& \int d\mu(\hat{U})  \left( \Tr [\hat{U} \hat{\rho} \hat{U}^\dagger \hat{H}] \right)^p  \\
&=&\sum_{\substack{j_1, \cdots j_p\\ k_1, \cdots, k_p}}  \rho_{j_1k_1}\cdots\rho_{j_pk_p} \sum_{\substack{\ell_1, \cdots \ell_p\\ i_1, \cdots, i_p}}
 H_{\ell_1i_1} \cdots H_{\ell_pi_p}  \nonumber \\
&\times& \int d\mu(\hat{U})  \; U_{i_1j_1} U_{i_2j_2} \cdots U_{i_pj_p} U^{\dagger}_{k_1\ell_1} U^{\dagger}_{k_2\ell_2} \cdots U^{\dagger}_{k_p\ell_p} \; , \nonumber 
\end{eqnarray}
where  $\rho_{jk}$, $H_{\ell i}$, $U_{ij}$ are matrix elements of $\hat{\rho}$, $\hat{H}$, $\hat{U}$ with respect to fixed basis
of the Hilbert space of the system~\cite{NOTA}.
The integral appearing on the last  term of~(\ref{integrale_momenti}) is widely known and admit solution~\cite{Weingarten1978,Creutz1978,Bars1980}
in terms of  the \emph{Weingarten functions} $C_{[\sigma]}$  for the unitary group~\cite{Weingarten1978,Collins2003}, i.e. 
\begin{eqnarray}
\nonumber
\int d\mu(\hat{U}) \; U_{i_1j_1} U_{i_2j_2} \cdots U_{i_pj_p} U^{\dagger}_{k_1\ell_1} U^{\dagger}_{k_2\ell_2} \cdots U^{\dagger}_{k_p\ell_p} \\
= \sum_{\tau, \sigma \in S_p} C_{[\sigma]} \prod_{a=1}^p  \delta_{i_a \ell_{\tau(a)} } \delta_{j_a k_{\tau\sigma(a)} } \; ,  
\label{integrale_U}
\end{eqnarray} 
with $S_p$ representing the  permutation group of $p$ elements.
A precise definition of the $C_{[\sigma]}$s and their properties will be given in Sec.~\ref{sec:weingarten}. Here we simply notice that 
replacing~(\ref{integrale_U}) into~(\ref{integrale_momenti}) leads to the expression 
\begin{equation}
\langle E^p\rangle =  
\sum_{\tau, \sigma \in S_p} C_{[\sigma]}  \; \rho{[\sigma\tau]}\;  H{[\tau]} \;,
\label{integrale_contrai}
\end{equation}
where for $\hat{\Theta}$ generic operator and $\sigma\in S_p$, we introduced  the functional
\begin{eqnarray}\label{defTHETASIGMA}  
\Theta{[\sigma]} &:=&   \sum_{k_1, \cdots, k_p} \Theta_{k_1 k_{\sigma(1)}} \Theta_{k_2 k_{\sigma(2)}} \cdots \Theta_{k_p k_{\sigma(p)}}\;.
\end{eqnarray} 
To get a more intuitive understanding of what is this product, we need to introduce an important property of a permutation $\sigma\in S_p$: the structure of its cycles, which is sufficient to specify its \emph{conjugacy class}
$[\sigma]$.
If the permutation $\sigma$ has a fixed point, i.e. if there exists an $i \in \{ 1,2, \cdots, p\}$ such that $\sigma(i) = i$, then we say that the permutation $\sigma$ has a cycle of length 1. A cycle of length 2 (i.e. 
a transposition) means that there exist two indices $i, j \in  \{ 1,2, \cdots, p\}$ such that $\sigma(i) = j$ and $\sigma(j) = i$. In general given 
 $\sigma\in S_p$ its conjugacy class $[\sigma]$ is uniquely identified via 
 the   correspondence 
 \begin{eqnarray} \label{correspondence} 
 [\sigma]  \longleftrightarrow  \{ \alpha^{(1)}_{[\sigma]}, \alpha^{(2)}_{[\sigma]}, \dots, \alpha^{(c[\sigma])}_{[\sigma]}\} \;, 
 \end{eqnarray}   
 with  
$c[\sigma]$ being the number  of independent cycles admitted by  $\sigma$, and with the positive  integers
 $\alpha^{(j)}_{[\sigma]}$'s representing instead  the lengths of such cycles organized in decreasing order, i.e. 
 $\alpha^{(j)}_{[\sigma]}\geq \alpha^{(j+1)}_{[\sigma]}$.
 Because each element in $\{1,2 , \cdots, p\}$ belongs to one and only one of the cycles of $\sigma$,  we must have that 
\begin{eqnarray} \sum_{j=1}^{c[\sigma]}  \alpha^{(j)}_{[\sigma]} = p\;, \label{die} \end{eqnarray} 
implying that the $\alpha^{(j)}_{[\sigma]}$s  provide a  proper partition  of the integer $p$ (represented graphically with a \emph{Young diagram}~\cite{Young1900}).
 The set of indices $\{1,2 , \cdots, p\}$ can thus be partitioned in $c[\sigma]$ subsets $\mathfrak{c}^{(i)}$, which each $\mathfrak{c}^{(i)}$ having  $\alpha^{(j)}_{[\sigma]}$ elements.
 We finally observe that two permutations $\sigma$ and $\sigma'$  will belong to the
 same   conjugacy class  if and only if we can 
 identify a third permutation that allows us to relate them via conjugation, i.e. 
  \begin{eqnarray}\label{conjugation} 
[\sigma] =[\sigma'] \qquad \Longleftrightarrow \qquad  \exists \tau\in S_p  | \quad  \sigma' = \tau \sigma \tau^{-1}\;. 
 \end{eqnarray} 
 The fact that each element $i \in  \{1,2 , \cdots, p\}$ belongs to exactly one cycle of $\sigma$ allows us to rewrite Eq.~(\ref{defTHETASIGMA}) as
 \begin{eqnarray} \label{questa}
 \Theta{[\sigma]} &=& \sum_{k_1, \cdots, k_p} \prod_{r=1}^{c[\sigma]}  \prod_{i \in \mathfrak{c}^{(r)}}  \Theta_{k_i k_{\sigma(i)}}
 \\  \nonumber 
& =& \prod_{r=1}^{c[\sigma]} \sum_{\{ k_i \}_{i \in \mathfrak{c}^{(r)}}} \prod_{i \in \mathfrak{c}^{(r)}}  \Theta_{k_i k_{\sigma(i)}}
 = 
  \prod_{r=1}^{c[\sigma]} \Tr[\hat{\Theta}^{\alpha^{(1)}_{[\sigma]}}] \; ,
 \end{eqnarray}
 where in the last identity we used the fact that the the $r$-th summation runs over a set of $\alpha^{(r)}_{[\sigma]}$ cyclical indices.
 Equation~(\ref{questa}) makes it clear that, as explicitly indicated by the notation, the terms $\Theta{[\sigma]}$  (as well as 
 the coefficients $C_{[\sigma]}$, see Eq.~(\ref{conjugation1}) below)
depend upon $\sigma$ only via its conjugacy class $[\sigma]$. 
 Replaced into Eq.~(\ref{integrale_contrai}), the identity~(\ref{questa}) also 
 implies that  $\langle E^p\rangle$ can be expressed as linear combination  of products of 
traces of powers of $\hat{\rho}$ and $\hat{H}$, i.e. explicitly 
\begin{equation}
\langle E^p\rangle =  \label{integrale_contrainew}
\sum_{\tau, \sigma \in S_p} C_{[\sigma]}  \;  \left(\prod_{i=1}^{c[\sigma\tau]}     \mbox{Tr}[\hat{\rho}^{\alpha^{(i)}_{[\sigma\tau]}}] \right)
\left( \prod_{j=1}^{c[\tau]}     \mbox{Tr}[\hat{H}^{\alpha^{(j)}_{[\tau]}}] \right)\;.
\end{equation}
Equation~(\ref{integrale_contrainew}) is the starting point of our analysis: we notice incidentally that it confirms 
by virtue of Specht's theorem~\cite{Specht1940}, that $P(E| \hat{\rho}; \hat{H})$ (and hence  its moments) depends upon 
 $\hat{\rho}$ and $\hat{H}$ only through their spectra.   
In particular 
for the cases  $p=1,2$ we get
\begin{eqnarray} \label{firstmoment} 
\langle E\rangle &=& C_{[1]} \mbox{Tr}[ \hat{H}]    \;, \\ \label{second} 
\langle E^2\rangle &=& C_{[1,1]} (\mbox{Tr}[ \hat{H}] )^2 + C_{[2]} (\Tr[\hat{\rho}^{2}]) (\mbox{Tr}[ \hat{H}] )^2 \label{varianza_in_funzione_degli_C} \\ \nonumber  
&&+ C_{[1,1]} (\mbox{Tr}[\hat{H}^{2}]) (\Tr[\hat{\rho}^{2}]) + C_{[2]} (\mbox{Tr}[\hat{H}^{2}]) \;, 
\end{eqnarray} 
which using the explicit values of the functions $C_{[\sigma]}$ reported in Tab.~\ref{tab1} leads to 
\begin{eqnarray}
\mu&:=& \langle E \rangle= \mbox{Tr}[ \hat{H}] /d \;, 
\label{energia_media_uniorb} \\ 
\label{varianza}
\Sigma^{(2)}&:=& \langle (E - \mu)^2 \rangle \\  \nonumber 
&=& \frac{1}{d^2-1} \left( \mbox{Tr}[\hat{H}^2] -\frac{(\mbox{Tr}[\hat{H}])^2}{d} \right) \left( \Tr[\hat{\rho}^{2}] - \frac{1}{d} \right)\!.
\end{eqnarray}

\subsection{Shifting the spectrum of $\hat{H}$} \label{sec:shift}

Equation~(\ref{firstmoment})  makes it clear that, irrespectively from the specific form of the input state of the system $\hat{\rho}$,  
we can enforce the distribution $P(E; \hat{\rho}, \hat{H})$ to have zero mean value by setting $\mbox{Tr}[ \hat{H}]  = 0$
 via a rigid shift of the associated spectrum. In particular setting $E_0=\mu$ in Eq.~(\ref{SHIFT}) we can write
 \begin{eqnarray} \label{SHIFTnewimpo1} 
P(E| \hat{\rho}; \hat{H})  = P\left(E-\mu \Big| \hat{\rho}; \Delta\hat{H} \right)\;,
\end{eqnarray} 
while, setting 
 $E_0=\mu$ and $\Delta_0=1$ in  Eq.~(\ref{SHIFTnew}) we get
\begin{eqnarray} \label{SHIFTnewimpo} 
P(E| \hat{\rho}; \hat{H}) 
= P\left(E-\mu \Big| \Delta \hat{\rho}; \Delta\hat{H} \right)\;, 
\end{eqnarray} 
where 
\begin{eqnarray} \label{TRACELESS} 
 \Delta \hat{\rho} :=  \hat{\rho}-\frac{{\hat{\openone}}}{d} \;, \qquad \Delta \hat{H} :=  \hat{H}-\mu {\hat{\openone}}\;,
\end{eqnarray} 
are zero-trace operators.
Following the  derivation of the previous section we can then use Eqs.~(\ref{SHIFTnewimpo1}) and~(\ref{SHIFTnewimpo}) 
to write the central moments of $P(E| \hat{\rho}; \hat{H})$ in the following equivalent forms 
\begin{eqnarray}
 \Sigma^{p} =  \langle (E -\mu)^p \rangle  &=& 
\sum_{\tau, \sigma \in S_p} C_{[\sigma]}  \; \rho{[\sigma\tau]}\;  \Delta H{[\tau]} \label{integrale_contrainuovo1}  
\\ 
 \label{integrale_contrainuovo}  
&=&\sum_{\tau, \sigma \in S_p} C_{[\sigma\tau^{-1}]}  \; \Delta \rho{[\sigma]}\;  \Delta H{[\tau]}\;, 
\end{eqnarray}
where $\Delta \rho{[\sigma]}$ and $\Delta H{[\tau]}$ are now the functional~(\ref{questa}) associated with the operators~(\ref{TRACELESS}), and where
in Eq.~(\ref{integrale_contrainuovo})   we changed the summation variable via the
introduction of the inverse $\tau^{-1}$ of the permutation  $\tau$. 
Both Eqs.~(\ref{integrale_contrainuovo1}) and  ~(\ref{integrale_contrainuovo}) offer us a huge simplification in the analysis of the problem,  with the first 
having an important application in the special case of  pure input states.
In particular we notice that since $\Delta \hat{H}$ is a traceless operators, in both these expressions we can 
restrict the summation  on $\tau$ by only including those permutations which have 
 no fixed points  (i.e. $\alpha^{(j)}_{[\tau]} \geq 2$ for all $j$).  Such elements define the $\emph{derangement}$ subset $S_p^{D}$ of $S_p$ 
   and by construction can have 
 at most $\lfloor p/2 \rfloor$ cycles, or equivalently 
 \begin{eqnarray} 
 \tau \in S_p^{D} \qquad \Longrightarrow \qquad  \lvert \tau \lvert \geq p-\lfloor p/2 \rfloor \;, \label{DER}
 \end{eqnarray} where 
 \begin{eqnarray}\lvert \tau \rvert =  p - c[\tau] \label{IMPO1} \;,
 \end{eqnarray} 
is the minimal number of transpositions (i.e. cycles of length 2)  $\tau$ is a product of
 (indeed, if there are more than $p/2$ cycles in the permutation $\tau$, there must be at least one element $i$ of the set  $\{ 1, 2, \cdots, p\}$ for which $\tau(i)=i$).
In the case of Eq.~(\ref{integrale_contrainuovo}) a similar simplification can also be enforced for the summation over $\sigma$, as again here one deals with a traceless operator
$\Delta \hat{\rho}$.
To summarise 
 the following selection rules hold:  
\begin{eqnarray}
\tau \notin S_p^{D}  \quad &\Longrightarrow& \quad  \Delta H[\tau] =  \prod_{j=1}^{c[\tau]}     \mbox{Tr}[\Delta\hat{H}^{\alpha^{(j)}_{[\tau]}}]  = 0 \;, 
\label{fixed_points_eq_01} \\ 
\sigma \notin S_p^{D}  \quad &\Longrightarrow& \quad  \Delta \rho[\sigma] =  \prod_{j=1}^{c[\sigma]}     \mbox{Tr}[\Delta\hat{\rho}^{\alpha^{(j)}_{[\sigma]}}]  = 0 \;, 
\label{fixed_points_eq_0}
\end{eqnarray}
so that we can rewrite  Eqs.~(\ref{integrale_contrainuovo1}) and ~(\ref{integrale_contrainuovo})  as  
\begin{eqnarray} 
\Sigma^{(p)} &=&\sum_{\sigma \in S_p}  \label{simplyp1}
\sum_{\tau \in S_p^{D}} C_{[\sigma]}  \;   \rho{[\sigma\tau]}\;  \Delta H{[\tau]} \nonumber  \\ 
&=&\sum_{\sigma,\tau \in S_p^{D}} C_{[\sigma\tau^{-1}]}  \;  \Delta \rho{[\sigma]}\;  \Delta H{[\tau]}   \;, \label{simplyp} 
\end{eqnarray} 
respectively.

		\begin{center} 		
			\begin{table*}[t!]
\begin{tabular}{|cc|}
\hline    
 \multicolumn{2}{|c|}{
 \begin{tabular}{cc||cc}
  & & \\
$p=1:$    &   \qquad \qquad  ${C_{[1]}=\frac{1}{d}}$ \qquad\qquad  &
 \qquad  $p=2:$ 
  \qquad & \qquad	$
 {C_{[{2}]}= -\frac{1}{d(d^2-1)}} \;,  \qquad  
 C_{[{1,1}]}
 = \frac{1}{d^2-1}\, ,$
   \qquad \\
 & &
 \end{tabular} } \\
 \hline  \hline  & \\ 
$\qquad p=3:$ & 
$C_{[3]}=\frac{2}{(d^2 -  4) (d^2 - 1) d }, \qquad 
	 C_{[{2,1}]}=-\frac{1 }{(d^2 -  4) (d^2 - 1) },
	\qquad
	C_{[1,1,1]} 
	=\frac{d^2 -   2}{(d^2 -  4) (d^2 - 1) d }\;,$
	\\ &  \\     \hline  \hline  & \\
$\qquad p=4:$	 & 
\quad $C_{[{4}]}=-\frac{5}{(d^2- 9) (d^2 - 4) (d^2 - 1) d}, 
	\quad
	C_{[{3,1}]}=\frac{2d^2 - 3}{(d^2- 9) (d^2 - 4) (d^2 - 1) d^2}, \quad 
{C_{[{2,2}]} 
	=\frac{d^2 +  6}{(d^2- 9) (d^2 - 4) (d^2 - 1) d^2} 
}\,,$ \\  & \\    &
$C_{[{2,1,1}]}
=-\frac{d^2-4}{(d^2- 9) (d^2 - 4)(d^2 - 1) d}
\,, \quad {C_{[{1,1,1,1}]}
	=\frac{d^4 - 8 d^2 +   6}{(d^2- 9) (d^2 - 4) (d^2 - 1) d^2}}\;.$\\ & \\  \hline
\end{tabular}
				\caption{List of the first 
				{Weingarten functions} $C_{[\sigma]}$ expressed in terms of the dimension $d$ of the system Hilbert space. Data adapted  from Ref.~\cite{Zuber2016}. 
				\label{tab1}} 
			\end{table*}
		\end{center} 	
 
\subsection{Combinatorial coefficients}
\label{sec:weingarten}
The Weingarten functions $C_{[\sigma]}$ introduced in Eq.~(\ref{integrale_U}) play a fundamental role in the analysis 
of the moments~(\ref{simplyp}). 
  These terms are in general difficult to compute but a formal expression for them is provided by  the formula \cite{FultonHarris2004,Noaves2014}
 \begin{equation}
C_{[\sigma]} = \frac{1}{p!^2}\sum_{\substack{\lambda\vdash p\\ c[\lambda]\leq d}}
\frac{ c_1[\sigma]  d_\lambda^2}{s_{\lambda, p}(1^p)}\; ,
\label{Weingarten}
\end{equation}
where the sum runs over all the irreducible representations $\lambda$
of the permutation group $S_p$ which are generated by partitions of $d$ which have at most $p$ elements, $d_\lambda$ is the dimension of the representation $\lambda$, $c_1[\sigma]$ is the number of cycles of length $1$ of $\sigma$ ~\cite{NOTA1} (a quantity sometimes called the \emph{character} of $\sigma$),  are the characters, and  $s_{\lambda,p}(x_1, \dots, x_p)$ are the Schur polynomials in $p$ variables (so that $s_{\lambda, p}(1^p) \equiv s_{\lambda, p}(1, \dots, 1)$ is the dimension of the representation of the unitary group $\mathbf{U}(d)$ which corresponds to $\lambda$ via the Schur-Weyl duality). 
 For $p$ small  
a list of the values of the $C_{[\sigma]}$ is reported in Tab.~\ref{tab1}.
As the number of irreducible representations of $S_p$ increases with $p$, the exact computation of these factor becomes rapidly a computationally infeasible task.
For the purposes of the present work, however  we do not need to compute the coefficients, indeed we just need to use the following known properties.

First of all, since  $C_{[\sigma]}$ is a functional of the conjugacy class, exploiting Eq.~(\ref{conjugation}) 
we can write 
  \begin{eqnarray}\label{conjugation1} 
C_{[\tau \sigma \tau^{-1}]}  =C_{[\sigma]}\;, 
 \end{eqnarray} 
 for all $\tau,\sigma\in S_p$. 
Second,  the overall sign of $C_{[\sigma]}$ is the sign of the permutation $\sigma$, namely
\begin{equation}  
\label{sign_of_C} \mbox{Sign}[ C_{[\sigma]} ] = 
(-1)^{\lvert \sigma \rvert} = (-1)^{p - c[\sigma]} \; ,
\end{equation}
where $\lvert \sigma \rvert$   is the minimal number of transpositions $\sigma$ 
is a product of, see Eq.~(\ref{IMPO1}).
Finally  we shall use the fact that for large enough $d$ the following asymptotic behaviour holds true
\begin{equation}
|C_{[\sigma]}| = \frac{ \mathcal{C}_{[\sigma]} }{d^{p+|\sigma|}} + \mathcal{O}\left( \frac{1}{d^{p+|\sigma|+2}} \right) \; ,
\label{C_sigmagenerico}
\end{equation}
where the integer number $\mathcal{C}_{[\sigma]}$ -- sometimes called the \emph{M\"obius function} of the permutation $\sigma$ -- is equal to the product
\begin{equation}
\mathcal{C}_{[\sigma]} := \prod_{i=1}^{c[\sigma]} Cat_{{\alpha^{(i)}_{[\sigma]}}} \; ,
\end{equation}
with the $\alpha^{(j)}_{[\sigma]}$'s  defined by the correspondence~(\ref{correspondence}), and with 
\begin{eqnarray} \label{catalan} 
Cat_n := \frac{1}{n+1} {2n \choose n}\;,\end{eqnarray} 
being the  $n$-th Catalan number. 
In particular 
 for the identical permutation this implies 
\begin{equation} 
|C_{[1^p]} |= \frac{1}{d^p} + \mathcal{O}\left( \frac{1}{d^{p+2}} \right) \; .
\label{C_1}
\end{equation}
The scaling~(\ref{C_sigmagenerico}) can be derived e.g. from Ref.~\cite{Collins2017}  where 
Collins and Matsumoto~proved that when 
\begin{eqnarray} p \leq   \left\lfloor (d / \sqrt{6})^{4/7} \right\rfloor \;, 
 \label{CONDDD} \end{eqnarray}  one has
\begin{equation}
\frac{\mathcal{C}_{[\sigma]}}{d^{p + \lvert \sigma \rvert}} \left( 1 - \frac{p-1}{d^2} \right)^{-1} \leq
|C_{[\sigma]}| \leq
\frac{\mathcal{C}_{[\sigma]}}{d^{p + \lvert \sigma \rvert}} 
\left( 1 - \frac{6p^{7/2}}{d^2} \right)^{-1} \; ,
\label{bounds_su_C}
\end{equation}
which via some simple algebraic manipulation can be casted in a weaker, but sometimes more useful  form: 
\begin{eqnarray}
\frac{\mathcal{C}_{[\sigma]}}{(d^2 - 1)^{p/2} d^{\rvert \sigma \lvert}} \left( 1 + \frac{p/2-1}{d^2} \right) \leq
|C_{[\sigma]} | \leq
\frac{\mathcal{C}_{[\sigma]}}{(d^2 - 1)^{p/2} d^{\rvert \sigma \lvert}}
\left( 1 - \frac{6p^{7/2}}{d^2} \right)^{-1} \; .
\label{bounds_su_C_var}
\end{eqnarray}

\section{Bounding the trace  terms}
\label{sec:normaH}

Here we establish some useful relations that allow us to  bound  
 the terms $\Delta H[\tau]$ and $\Delta \rho[\sigma]$ entering in Eqs.~(\ref{simplyp1}) and (\ref{simplyp}) and which 
 will allow us to identify the necessary conditions on $\hat{H}$ that are needed to prove the Gaussian approximation~(\ref{intro_momenti_gaussiani}).
 For this purpose we shall relay on the properties of the derangement set $S_p^{D}$ and on the inequality~\cite{Kittaneh1985}
\begin{equation}
\label{ineqL2_rho}
| \mbox{Tr}[\hat{\Theta}^{q_2}] |^{1/q_2} \leq \| \hat{\Theta}\|_{q_2} \leq \| \hat{\Theta}\|_{q_1} \;, \quad \forall q_2\geq q_1>0\;, 
\end{equation} 
where for 
$q>0$
\begin{eqnarray} 
\| \hat{\Theta}\|_q: = \mbox{Tr}[ | \hat{\Theta}|^q]^{1/q}\label{SHATTEN} \;, 
\end{eqnarray}  is the $q$-th Shatten norm of the operator $\hat{\Theta}$.

\subsection{A useful inequality} \label{Sec:useful_ineq} 

 To begin with let consider the special subset  $S_p^{D*}$  of $S_p^{D}$ formed by those derangements $\tau$ 
 which can be decomposed into $p/2$ cycles of length 2. 
 If $p$ is even   there   are  $(p-1)!!$ of such  elements, i.e. 
 \begin{eqnarray} \label{DIMSPD*}
 |S_p^{D*}| = (p-1)!!  \;, \qquad \mbox{($p$ even)}\;. 
 \end{eqnarray} 
On the contrary
 if $p$ is odd  there are no such permutations (at least one cycle must be of odd length): in this case
 we identify $S_p^{D*}$ with the empty set
\begin{eqnarray} 
S_p^{D*}= \varnothing \;, \qquad \mbox{($p$ odd)}\;. 
\end{eqnarray} 
By definition  the elements of $S_p^{D*}$ verify the condition 
  $\alpha^{(j)}_{[\tau]}=2$ for all $j$ and thus, 
by explicit computation, they fulfil the identity 
 \begin{eqnarray} \label{df12} 
 \Theta[\tau] =  (\mbox{Tr}[\hat{\Theta}^{2}])^{p/2} = \| \hat{\Theta}\|_2^{p}\;,   \qquad \forall \tau\in  S_p^{D*} \;,
 \end{eqnarray} 
 for all operators $\hat{\Theta}$. 
 Now let  $\tau \in S_p^{D}/S_p^{D*}$ a derangement which cannot be decomposed into $p/2$ cycles of length 2:
 for such permutations we can prove that $\| \hat{\Theta}\|_2^{p}$ provides an upper bound for the associated value  $|\Theta[\tau]|$, i.e. 
\begin{eqnarray} \label{THEBOUND1} 
{ |\Theta[\tau]|}  \leq \| \hat{\Theta}\|_2^{p} \times \left\{ \begin{array}{ll}  \eta_{\hat{\Theta}}  \;,  &\quad  \mbox{\rm($p$ even)}\;, \\ \\
 \sqrt{\eta_{\hat{\Theta}}} \;,  &\quad  \mbox{\rm{($p$ odd)}}\;,
  \end{array} \right.
\end{eqnarray} 
 where 
 \begin{eqnarray}  \label{FONDterm} 
\eta_{\hat{\Theta}}  :=\left( \frac{\| \hat{\Theta}\|_3}{\| \hat{\Theta}\|_2}\right)^6 = \frac{(\mbox{Tr}[|\hat{\Theta}|^{3}])^{2}}{ (\mbox{Tr}[\hat{\Theta}^{2}])^3 } \;,
\end{eqnarray} 
is a  functional of $\hat{\Theta}$ which due to~(\ref{ineqL2_rho})
fulfils  the inequality~\cite{NOTA2} 
\begin{eqnarray}\frac{1}{d^6} \leq \eta_{\hat{\Theta}} \leq 1\;. \label{CONDNAT} \end{eqnarray} 
Equation~(\ref{THEBOUND1}) is a direct consequence of the following  more general observation:

\begin{prop}  \label{propo1} 
Given $\tau$ an element of the derangement set $S_p^{D}$ we have 
\begin{eqnarray} \label{THEBOUND} 
| \Theta[\tau] | &\leq& \eta_{\hat{\Theta}}^{ |\tau|-p/2}   (\mbox{\rm{Tr}}[\hat{\Theta}^{2}])^{p/2}  =\eta_{\hat{\Theta}}^{ |\tau|-p/2}   \| \hat{\Theta}\|_2^{p}  \;,
\end{eqnarray} 
with  $\eta_{\hat{\Theta}}$ the parameter defined in Eq.~(\ref{FONDterm}).
\end{prop} 

\begin{proof}
The bound~(\ref{THEBOUND}) is clearly fulfilled by $\tau\in S_p^{D*}$ due to Eq.~(\ref{df12}) and the fact that for such permutations $c[\tau]=p/2=|\tau|$, e.g. see Eq.~(\ref{IMPO1}).

Now let  $\tau \in S_p^{D}/S_p^{D*}$ a derangement which cannot be decomposed into $p/2$ cycles of length 2. 
Since $\tau$ is an element of $S_p^{D}$ 
all its coefficients 
 $\{ \alpha^{(1)}_{[\tau]}, \alpha^{(2)}_{[\tau]}, \dots, \alpha^{(c[\tau])}_{[\tau]}\}$ are guaranteed to be greater or equal to 2, i.e.
 \begin{eqnarray} 
\tau \in S_p^{D} \; \Longrightarrow \; \alpha^{(j)}_{[\tau]} \geq 2 \qquad \forall j\;.
 \end{eqnarray} 
 However since $\tau\notin S_{p}^{D*}$ then some those terms 
(say the first $K\geq 1$) must be  strictly larger than $2$, i.e. 
\begin{eqnarray}  \label{IPO} 
  \alpha^{(j)}_{[\tau]} \geq 3,  \qquad \forall 1\leq  j \leq K\;. 
\end{eqnarray} 
Invoking~(\ref{ineqL2_rho}) we can hence claim 
\begin{equation} 
\left|\mbox{Tr}[\hat{\Theta}^{\alpha^{(j)}_{[\tau]}}]  \right|  \leq (\mbox{Tr}[|\hat{\Theta}|^{3}])^{\alpha^{(j)}_{[\tau]}/3} =
\| \hat{\Theta}\|_3^{\alpha^{(j)}_{[\tau]}}\;, 
  \quad 
\forall  1 < j \leq K\;,
\end{equation} 
and 
\begin{equation} 
\left|\mbox{Tr}[\hat{\Theta}^{\alpha^{(j)}_{[\tau]}}]  \right|\leq \mbox{Tr}[\hat{\Theta}^{2}]^{{\alpha^{(j)}_{[\tau]}}/{2}} 
 =
\| \hat{\Theta}\|_2^{\alpha^{(j)}_{[\tau]}} \;, \quad 
\forall  K < j \leq c[\tau]\;.  
\end{equation} 
Observe also that 
since  $K$ can be expressed as  the number of cycles $c[\tau]$ minus the  number  $c_2[\tau]$ of cycles of length $2$, and since the latter is always smaller than or equal to the minimal number $|\tau|$ of transpositions
$\tau$ can be decomposed of, we can bound this quantity as follows 
\begin{eqnarray}
K = c[\tau] - c_2[\tau] \geq  c[\tau]   - |\tau| = p-2|\tau| \;, \label{defK}
\end{eqnarray}  
where in the last identity  we used~(\ref{IMPO1}).
From the above expressions we can hence establish  that  for all $\tau\in S_p^{D}/S_p^{D*}$ we have 
 \begin{eqnarray} 
| \Theta[\tau] | &=&   \left| \prod_{j=1}^{c[\tau]}     \mbox{Tr}[\hat{\Theta}^{\alpha^{(j)}_{[\tau]}}] \right| 
\leq  \| \hat{H}\|_3^{\sum_{j=1}^{K} \alpha^{(j)}_{[\tau]}} \;  \| \hat{\Theta}\|_2^{p-\sum_{j=1}^{K} \alpha^{(j)}_{[\tau]}} \nonumber 
\\
&=& 
\eta_{\hat{\Theta}}^{{\sum_{j=1}^{K} \alpha^{(j)}_{[\tau]}}/{6}} \| \hat{\Theta}\|_2^p
  \;,  \label{almostfinale} 
\end{eqnarray} 
where we  employ the definition (\ref{FONDterm}) and use  the normalization condition Eq.~(\ref{die}) to write 
\begin{eqnarray} 
\sum_{j=K+1}^{c[\tau]} \alpha^{(j)}_{[\tau]}=
 p - \sum_{j=1}^{K} \alpha^{(j)}_{[\tau]} \;. \end{eqnarray} 
 From Eqs.~(\ref{IPO}) and (\ref{defK}) 
we have also
\begin{eqnarray}
{\frac{\sum_{j=1}^{K} \alpha^{(j)}_{[\tau]}}{6}} \geq \frac{K}{2} \geq \frac{p}{2} - |\tau|\;. 
\end{eqnarray} 
Therefore we can claim
\begin{eqnarray}
\eta_{\hat{\Theta}}^{\frac{\sum_{j=1}^{K} \alpha^{(j)}_{[\tau]}}{6}} \leq \eta_{\hat{\Theta}}^{\frac{p}{2} - |\tau|}\;, 
\end{eqnarray} 
which replaced into Eq.~(\ref{almostfinale}) finally yields Eq.~(\ref{THEBOUND}). 
\end{proof} 

The derivation of (\ref{THEBOUND1}) from (\ref{THEBOUND}) finally proceed by observing that for $\tau\in S_p^{D}/S_p^{D*}$ the exponent of $\eta_{\hat{\Theta}}$ appearing in (\ref{THEBOUND}) 
 is always greater or equal to $1$ when $p$ even, and greater or equal to $1/2$ for $p$ odd, i.e. 
\begin{eqnarray} \label{THEBOUND2} 
|\tau|-p/2=p/2-c[\tau] \geq  \left\{ \begin{array}{ll} 1 \;,  &\quad  \mbox{\rm($p$ even)}\;, \\ \\
1/2 \;,  &\quad  \mbox{\rm{($p$ odd)}}\;.
 \end{array} \right.
\end{eqnarray}

\subsection{Condition upon the Hamiltonian} \label{CONDHAMSEC} 
Setting  $\hat{\Theta} = \Delta \hat{H}$ in Eqs.~(\ref{df12}) and~(\ref{THEBOUND}) 
of the previous section 
we get that  the contributions $\Delta H[\tau]$  entering 
in  Eqs.~(\ref{simplyp1}) and (\ref{simplyp})
are always upper bounded by  the quantity    $\|\Delta \hat{H}\|_2^p$,~i.e. explicitly 
\begin{eqnarray}\label{THEBOUND0} 
{ \Delta H[\tau]}  &=& \|\Delta \hat{H}\|_2^p \;, \qquad  \forall \tau \in S_p^{D*}\;, \\ \nonumber \\ 
{ |\Delta H[\tau]|}  &\leq& \| \Delta \hat{H}\|_2^{p} \times \left\{ \begin{array}{lll}  \eta_{\Delta \hat{H}}    &  \mbox{\rm($p$ even)}& \\
&& \forall \tau \in S_p^D/S_p^{D*}\;,  \\
 \sqrt{\eta_{\Delta \hat{H}}}   & \mbox{\rm{($p$ odd)}}&
  \end{array} \right. \nonumber \\
  \label{THEBOUND1H} 
\end{eqnarray} 
with 
 \begin{eqnarray}  \label{FONDtermH} 
\eta_{\Delta\hat{H}}  :=\left( \tfrac{\| \Delta\hat{H}\|_3}{\| \Delta\hat{H}\|_2}\right)^6  \in \left[ \frac{1}{d^6} ,1\right]\;. 
\end{eqnarray} 
As we shall see the fundamental ingredient  to prove
Eq.~(\ref{intro_momenti_gaussiani}) is to strengthen 
 the inequality~(\ref{THEBOUND1H}) to make sure that the contributions to $\Sigma^{(p)}$  associated with the derangements $\tau$ that do not belong to $S_p^{D*}$ are very much depressed with respect to  $\|\Delta \hat{H}\|_2^p$, i.e. to impose the constraint
 \begin{eqnarray}\label{THEBOUND1new}
\Delta H[\tau] \ll  \|\Delta \hat{H}\|_2^p\;, \qquad \forall \tau\in S_p^{D} / S_p^{D*} \;,
\end{eqnarray} 
or equivalently to have 
 \begin{eqnarray}\eta_{\Delta \hat{H}} \ll 1\;,\qquad \Longleftrightarrow \qquad  \| \Delta \hat{H}\|_3 \ll \|\Delta \hat{H}\|_2 \;. \label{cond_gauss_H} \end{eqnarray}
It should be clear that due to the fact that since  $\eta_{\Delta\hat{H}}$ is always greater than or equal to ${1}/{d^6}$, Eq.~(\ref{cond_gauss_H}) 
 can only be fulfilled when operating with large Hilbert space, i.e. 
\begin{equation}
\label{d_gg_1}
\eta_{\Delta \hat{H}} \ll 1 \qquad \Longrightarrow \qquad  d \gg 1 \; .
\end{equation}
Notice however that once the requirement $d\gg1$ is met, the regime~(\ref{cond_gauss_H}) 
is  easy  to achieve, as it 
can only fail in the special case of
  Hamiltonians that have  few eigenvalues much greater than all the others.

 \section{Central moments asymptotic expression   for  pure input states} 
\label{sec:outlook_statipuri}
  In this section we  present a proof of our main result that applies in the special scenario where the input state of the system is 
pure, i.e. 
  \begin{eqnarray} \label{PURESTATE} \hat{\rho} = \ket{\psi}\bra{\psi}\;. \end{eqnarray}
This  constraint  allows for some  simplifications that help in clarifying the derivation.
Indeed thanks to the assumption~(\ref{PURESTATE}) we have 
$\Tr[\hat{\rho}^n] = 1$ for every $n$, which leads to the useful identity  
\begin{eqnarray}  \label{PURE} 
\rho[\sigma ]= \prod_{i=1}^{c[\sigma]}     \mbox{Tr}[\hat{\rho}^{\alpha^{(i)}_{[\sigma]}}]  =1 \;, \quad \forall \sigma \in S_p\;.
\end{eqnarray}  
 Therefore invoking Eq.~(\ref{simplyp1}) we can  express the central moment $\Sigma^{(p)}$ of a pure input state,  as the product of a numerical factor that only depends
 on the geometrical properties of derangements, times a contribution that fully captures the dependence upon the system Hamiltonian, i.e. 
\begin{eqnarray}
 \Sigma^{(p)} &=&  \label{integrale_contrainew_pure}
\left( \sum_{ \sigma \in S_p} C_{[\sigma]}  \right) \;\left( \sum_{\tau  \in S^D_p} \Delta H[\tau] \right)  \;.
\end{eqnarray}
This quantity should be compared with the corresponding Gaussian expression~(\ref{GAUSMOM}) 
with the scaling factor 
 \begin{eqnarray} \label{GAUSMOMpure} 
  {\mathbf G}_{p} =
	 \frac{(p-1)!!
 }{(d(d+1))^{p/2}} \| \Delta \hat{H}\|_2^{p}\;,
 \end{eqnarray} 
 obtained from (\ref{SCALING}) using the fact that 
  for pure states Eq.~(\ref{varianza})  gives us 
 \begin{eqnarray}
\Sigma^{(2)} &=& \frac{\| \Delta \hat{H}\|_2^2
 }{d(d+1)}\;.
 \end{eqnarray} 
 Under the above premise in the next section we show that Eq.~(\ref{sezII_bbound_erroresusigma1}) applies with 
  \begin{eqnarray}
f_{\hat{H}}(d,p) &=& \left\{ \begin{array}{lr}  \tfrac{p(p-2)}{2d}  + 
  \eta_{\Delta\hat{H}}\left( \tfrac{ \left\lceil {p!}/{e} \right\rceil}{(p-1)!!}   - 1 \right),  & \mbox{($p$ even),}   \\  \\ 
\sqrt{\eta_{\Delta\hat{H}}}\;  \tfrac{\left\lceil {p!}/{e} \right\rceil }{(p-1)!!}\;,   &\mbox{($p$ odd),} 
\end{array} \right.
  \label{riusltatopuri} 
 \end{eqnarray}
 with $\eta_{\Delta\hat{H}}$ the coefficient defined in Eq.~(\ref{FONDtermH}).
 It is worth stressing that in this case, the validity of Eq.~(\ref{sezII_bbound_erroresusigma1}) is not subjected to any constraint 
 on $p$ and $d$, but of course $f_{\hat{H}}(d,p)$ can become  small  only if $d$ is much smaller than~$p$.

 \subsection{The geometric coefficient Eq~(\ref{integrale_contrainew_pure}) } \label{firstCOEF}
 
 A closed expression for the first factor on the r.h.s. of Eq.~(\ref{integrale_contrainew_pure}) is provided by the formula
\begin{eqnarray}
\sum_{ \sigma \in S_p} C_{[\sigma]}  &=& \frac{(d-1)!}{(p+d-1)!}= \frac{D(d,p)}{(d(d+1))^{p/2}} \;, \label{FORMULA} 
\end{eqnarray} 
with  
\begin{eqnarray} \label{GAMMABOUND} 
D(d,p)  &: =&  \frac{(d-1)!d^{p/2}(d+1)^{p/2}}{(p+d-1)!} \;, 
\end{eqnarray}
being a numerical coefficient that for $p\geq 2$ fulfils the inequalities
\begin{eqnarray} 
1\geq D(d,p) &\geq&  \left\{ \begin{array}{ll} 
 1- \tfrac{p(p-2)}{2d}  & \mbox{($p$ even),} \\ \\
  1- \tfrac{(p+1)(p-1)}{2d}  & \mbox{($p$ odd).} 
  \end{array} \right.
\end{eqnarray} 
An explicit derivation of Eq.~(\ref{FORMULA})  can be obtained by focusing on the trivial case where 
 the system Hamiltonian is proportional to the identity operator, i.e. 
 \begin{eqnarray} \label{ddf} 
\hat{H} = E_0\hat{\openone} \;, \qquad \Longrightarrow\qquad  \Delta \hat{H} = 0\;.
\end{eqnarray} 
In this case  in agreement with Eq.~(\ref{DELTA})  the terms~(\ref{integrale_contrainew_pure}) vanish leading to the following identities
\begin{eqnarray} 
0= \Sigma^{(p)}= \langle E^p\rangle - E_0^p\;,
\end{eqnarray}  
where we used the fact that now $\mu = E_0$.
Invoking hence Eqs.~(\ref{integrale_contrai}) and~(\ref{PURE}) we can then write 
\begin{eqnarray} \label{Fin11} 
E_0^p =\langle E^p\rangle &=&\left( \sum_{ \sigma \in S_p} C_{[\sigma]}  \right) \;\left(  \sum_{\tau  \in S_p}  H[\tau] \right)\;.
\end{eqnarray} 
Our next step is to compute $\sum_{\tau  \in S_p}  H[\tau]$. To begin with we  write  the summation over $\tau$ by grouping together those permutations which are characterized by the same value of $|\tau|=k$, decomposing  $S_p$ into the subsets 
 \begin{eqnarray} \label{defspk} 
 S_p(k) := \{ \tau \in S_p \; \mbox{s.t.} \;  |\tau| =k\} \;,
 \end{eqnarray} 
that is 
\begin{eqnarray}\label{REP1} 
  \sum_{\tau  \in S_p}  H[\tau] =    \sum_{k=0}^p   \sum_{\tau  \in S_p(k)}  H[\tau]\;.
\end{eqnarray} 
Then we recall that the number of elements of $S_p(k)$ corresponds to the number of permutations having exaclty $p-k$ cycles, i.e.
 to  the Stirling numbers of the first kind $\left[ {p \atop p-k} \right]$~\cite{ConcreteMathematics}.
Observe next that from Eq.~(\ref{ddf}) and ~(\ref{IMPO1}) we have 
\begin{eqnarray} 
 H[\tau] = E_0^p \; \openone[\tau] = E_0^p \;  d^{c[\tau]} = E_0^p \; d^{p- |\tau|}  \;,
\end{eqnarray} 
which replaced into (\ref{REP1}) gives 
\begin{equation}\label{fin11} 
  \sum_{\tau  \in S_p}  H[\tau] =   E_0^p   \sum_{k=0}^p  d^{p- k}  \left[ {p \atop p-k} \right] =   E_0^p  \frac{(p+d-1)!}{(d-1)!}  \;,
\end{equation} 
where we invoked  the useful property~\cite{ConcreteMathematics} 
 \begin{equation}
\sum_{k=0}^p  d^{p-k} \left[ {p \atop p-k} \right]  = \frac{(p+d-1)!}{(d-1)!}  \; .
\label{somma_stirling} 
\end{equation}
Equation~(\ref{FORMULA}) finally follows by substituting (\ref{fin11}) into (\ref{Fin11}) and performing some trivial simplifications. 
 
 We conclude by noticing that the bounds~(\ref{GAMMABOUND})  on $D(d,p)$ can be established by 
 observing that 
  it can be expressed as 
\begin{eqnarray} 
D(d,p) &=& \left\{ \begin{array}{ll}
\left(  \tfrac{d}{d+2} \tfrac{d+1}{d+3}\right) 
 \cdots \left( \tfrac{d}{d+p-2} \tfrac{d+1}{d+p-1}\right)  & \mbox{($p$ even),} \\
 \\  \nonumber 
 \left(  \tfrac{d}{d+2} \tfrac{d+1}{d+3}\right) 
 \cdots \left( \tfrac{d}{d+p-3} \tfrac{d+1}{d+p-2}\right) \tfrac{\sqrt{d(d+1)}}{d+p-1}  & \mbox{($p$ odd),} \\
 \end{array} 
 \right.
\end{eqnarray} 
which immediately reveals that  it is always smaller than or equal to 1 for all  $p\geq 2$. 
The lower bound instead follows by using  the inequalities $(1+x)^\alpha > 1 + \alpha x$ and then $\prod_i (1-x_i) > 1 - \sum_i x_i$,
and noticing that  in case of $p$ even we have 
\begin{eqnarray} 
D(d,p)  &\geq&  \nonumber 
  \left(  \tfrac{d}{d+2} \tfrac{d}{d+2}\right) \cdots \left( \tfrac{d}{d+p-2} \tfrac{d}{d+p-2}\right) \\   \nonumber 
&=&  \prod_{i=0}^{p/2 - 1} \left( 1 + \tfrac{2i}{d} \right)^{-2} \geq  \prod_{i=0}^{p/2 - 1} \left( 1 - \tfrac{4i}{d} \right) \\ 
&\geq&   1 - \sum_{i=0}^{p/2 - 1}  \tfrac{4i}{d} 
= 
 1- \tfrac{p(p-2)}{2d} \;.  \label{ineqevene} 
\end{eqnarray} 
For  $p$ odd we have instead 
\begin{eqnarray} 
&&D(d,p)  =  \nonumber \left(  \tfrac{d}{d+2} \tfrac{d+1}{d+3}\right) 
 \cdots \left( \tfrac{d}{d+p-3} \tfrac{d+1}{d+p-2}\right)  \tfrac{\sqrt{d(d+1)}}{d+p-1}
 \\
&&=  \nonumber \left(  \tfrac{d}{d+2} \tfrac{d+1}{d+3}\right) 
 \cdots \left( \tfrac{d}{d+p-3} \tfrac{d+1}{d+p-2}\right)  \left( \tfrac{d}{d+p-1} \tfrac{d+1}{d+p}\right) \tfrac{d+p}{\sqrt{d(d+1)}} \\
&&=    \Gamma(d,p+1) \tfrac{d+p}{\sqrt{d(d+1)}}  
\nonumber \\
&&\geq  \Gamma(d,p+1)\geq 1- \tfrac{(p+1)(p -1)}{2 d}\;,
\end{eqnarray} 
where the first inequality is obtained by invoking~(\ref{ineqevene}). 

 \subsection{The Hamiltonian coefficient Eq~(\ref{integrale_contrainew_pure}) } \label{SECCOEF}

In studying  the Hamiltonian contribution to Eq.~(\ref{integrale_contrainew_pure}) 
for $p$ even we split the summation in two parts writing 
\begin{eqnarray}
 \sum_{\tau  \in S^D_p}  \Delta H[\tau] &=& \sum_{\tau  \in S^{D*}_p}  \Delta H[\tau]  +
 \sum_{\tau  \in S^D_p/S^{D*}_p}  \Delta H[\tau]  \nonumber \\
 &=& (p-1)!! \|\Delta \hat{H}\|_2^{p}  +  \sum_{\tau  \in S^D_p/S^{D*}_p}  \Delta H[\tau]  \;, 
  \nonumber \\
\end{eqnarray} 
where we used the identities (\ref{THEBOUND0}) and (\ref{DIMSPD*}) to compute the first contribution. 
The second contribution instead can be bounded by invoking 
Eq.~(\ref{THEBOUND1H}) to write 
\begin{eqnarray}
  \left| \sum_{\tau  \in S_p^D/ S^{D*}_p} \Delta H[\tau] \right| 
 &\leq&  \eta_{\Delta\hat{H}}  \|\Delta \hat{H}\|_2^{p} \left| S^{D}_p/S^{D*}_p\right|  \\ \nonumber
  &\leq&  \eta_{\Delta\hat{H}}  \|\Delta \hat{H}\|_2^{p}  \left( \left\lceil {p!}/{e} \right\rceil - (p-1)!! \right) \;,
\end{eqnarray} 
 where 
 \begin{equation} \label{stima} 
 \left| S^{D}_p/S^{D*}_p\right| := \left| S^{D}_p\right| - \left| S^{D*}_p\right| =  \left| S^{D}_p\right| - (p-1)!!\;, 
 \end{equation}  is the number of permutation in $S^{D}_p/S^{D*}_p$ which we bounded by exploiting the fact that 
 the number of total elements of the derangement set  $S_p^{D}$ is \cite{Whitwort1878}
\begin{equation}
\left| S^{D}_p \right|  = \left\lfloor \; {p!}/{e} + {1}/{2} \right\rfloor \leq \left\lceil {p!}/{e} \right\rceil \; .
\label{number_of_derangements}
\end{equation}
 For $p$ odd instead $S_p^{D*}$ is the empty set and  we get
 \begin{eqnarray}
 \left| \sum_{\tau  \in S^D_p} \Delta H[\tau] \right| &\leq&  \sqrt{\eta_{\Delta\hat{H}}}  \|\Delta \hat{H}\|_2^{p} \left| S^{D}_p\right| \nonumber \\
&\leq& \sqrt{\eta_{\Delta\hat{H}}}  \|\Delta \hat{H}\|_2^{p} \left\lceil {p!}/{e} \right\rceil\;. 
\end{eqnarray}

 \subsection{Derivation of Eq.~(\ref{riusltatopuri})} \label{final} 

Exploiting the results of the previous sections we can conclude that for $p$ even the following inequalities apply
  \begin{eqnarray}
&&|\Sigma^{(p)}- \Sigma_G^{(p)} | = \left|\tfrac{D(d,p)}{(d(d+1))^{p/2}} \sum_{\tau  \in S^D_p} \Delta H[\tau]   - \Sigma_G^{(p)} \right|\nonumber \\
&&\quad = \left|D(d,p)\left( \tfrac{(p-1)!! \|\Delta \hat{H}\|_2^{p} }{(d(d+1))^{p/2}} + 
\tfrac{ \sum_{\tau  \in S^D_p/ S^{D*}_p} \Delta H[\tau] }{(d(d+1))^{p/2}} \right)  - \Sigma_G^{(p)} \right| \nonumber \\
&&\quad = \left|(D(d,p)-1) \Sigma_G^{(p)}  + D(d,p)
\tfrac{ \sum_{\tau  \in S^D_p/ S^{D*}_p} \Delta H[\tau] }{(d(d+1))^{p/2}}  \right| \nonumber \\
&&\quad \leq  \left| D(d,p)-1\right| \Sigma_G^{(p)}  + \tfrac{D(d,p)}{(d(d+1))^{p/2}}
 \left| \sum_{\tau  \in S^D_p/ S^{D*}_p} \Delta H[\tau] \right|  \nonumber \\
 &&\quad \leq  \tfrac{p(p-2)}{2d}  \Sigma_G^{(p)}  + 
  \eta_{\Delta\hat{H}} \tfrac{ \|\Delta \hat{H}\|_2^{p} \left(\left\lceil {p!}/{e} \right\rceil - (p-1)!! \right) }{(d(d+1))^{p/2}}    \nonumber \\
 &&\quad = \left[ \tfrac{p(p-2)}{2d}  + 
  \eta_{\Delta\hat{H}}\left( \tfrac{ \left\lceil {p!}/{e} \right\rceil}{(p-1)!!}   - 1 \right)  \right]    {\mathbf G}_{p} \;,
 \end{eqnarray} 
 in agreement with identifying the function $f_{\hat{H}}(d,p)$ of  Eq.~(\ref{sezII_bbound_erroresusigma1})  as  anticipated in Eq.~(\ref{riusltatopuri}). Similarly for $p$ odd  we can write 
  \begin{eqnarray}
&&|\Sigma^{(p)} | = \left|\frac{D(d,p)}{(d(d+1))^{p/2}} \sum_{\tau  \in S^D_p} \Delta H[\tau]   \right| \\\nonumber
&&\qquad \leq \sqrt{\eta_{\Delta\hat{H}}}\; \frac{ \|\Delta \hat{H}\|_2^{p} }{(d(d+1))^{p/2}} \left\lceil {p!}/{e} \right\rceil
  \leq \sqrt{\eta_{\Delta\hat{H}}}\;  \tfrac{\left\lceil {p!}/{e} \right\rceil }{(p-1)!!} \; {\mathbf G}_{p} \;, 
 \end{eqnarray}
 which again corresponds to set $f_{\hat{H}}(d,p)$ of  Eq.~(\ref{sezII_bbound_erroresusigma1}) as in  Eq.~(\ref{riusltatopuri}).

\section{Asymptotic expression for  the central moments for arbitrary (non necessarily pure) input states}

\label{sec:bound_momenti}

Here we present the general proof that, for arbitrary (non necessarily pure) states,  under the assumption (\ref{CONDONP}) 
 the centred moments of the distribution~$P(E| \hat{\rho}; \hat{H})$ are well approximated by the Gaussian relations~(\ref{GAUSMOM}) whose
 scaling factors (\ref{SCALING}) can be conveniently expressed as
 \begin{eqnarray} \label{SCALINGNONPURE} 
 {\mathbf G}_{p} =\frac{(p-1)!!}{(d^2-1)^{p/2}}  \| \Delta \hat{H}\|_2^p \;  \| \Delta \hat{\rho}\|_2^p \;, 
 \end{eqnarray} 
 upon rewriting Eq.~(\ref{varianza}) in the compact form 
 \begin{eqnarray}
\label{varianzaCOMP}
\Sigma^{(2)}=
 \frac{ \| \Delta \hat{H}\|_2^2 \;  \| \Delta \hat{\rho}\|_2^2}{d^2-1} \;.
\end{eqnarray}
 In particular we shall show that under the condition 
  \begin{eqnarray} p \leq  \min \left\{ \left\lfloor \sqrt{d} \right\rfloor  , \left\lfloor (d / \sqrt{6})^{4/7} \right\rfloor \right\}\;,\label{CONDONP}
 \end{eqnarray}
   the inequalities~(\ref{sezII_bbound_erroresusigma1})
 hold true 
 with the function $f_{\hat{H}}(d,p)$ defined as 
		\begin{multline}
 f_{\hat{H}}(d,p)  :=  \left( 1 - \tfrac{6p^{7/2}}{d^2} \right)^{-1} \Big[  \tfrac{6p^{7/2}}{d^{2}} + \;  \tfrac{p^2 Cat_p}{d} \\
  + \eta_{\Delta \hat{H}} \left( \tfrac{\left\lceil {p!}/{e} \right\rceil }{(p-1)!!} -  1\right)
 \;Cat_p ( 1+ \tfrac{p^2}{d}) \Big] \;, 
	\label{sezII_bbound_erroresusigma}
	\end{multline}
	for $p$ even and $ \eta_{\Delta \hat{H}}$ defined as in Eq.~(\ref{FONDtermH}), and
\begin{multline}
	f_{\hat{H}}(d,p)  :=\sqrt{\eta_{\Delta \hat{H}}}  \left( 1 - \tfrac{6p^{7/2}}{d^2} \right)^{-1} 
    \left( \tfrac{\left\lceil {p!}/{e} \right\rceil }{(p-1)!!} -  1\right) \\ 
 \times Cat_p ( 1+ \tfrac{p^2}{d})  
	\label{intro_tauscartati_simpler}\;,
	\end{multline}
	for $p$ odd. 
Notice that  weaker, but  possibly more friendly, expressions can also be obtained 
by  using   the  following (generous) upper bound for the Catalan numbers~\cite{Sprugnoli1990}
\begin{equation}
{Cat}_p < \frac{4^p}{\sqrt{\pi}p^{3/2}} \; ,
\label{bound_catalan}
\end{equation}
and observing that 
under the condition $p^2 < d$, provided that $1 - 6d^{-1/4} > 0$ (i.e., $d>6^4=1296$) we can also write \begin{eqnarray}
\left( 1 - \frac{6p^{7/2}}{d^2} \right)^{-1} \leq \frac{1}{1 - 6d^{-1/4}} \; ,
\label{bound_matsu}
\end{eqnarray}
which   does not depend on $p$.

 \subsection{Case $p$ even}
\label{sec:moments_even}
In Sec.~\ref{CONDHAMSEC} we noticed that if $p$ is even the special set $S^{D*}_p$  is not empty and it is expected to provide the most relevant 
contributions in the summation over $\tau$ that defines $\Sigma^{(p)}$. Accordingly we start  splitting  the summation~(\ref{simplyp}) in two parts, i.e.
 \begin{eqnarray} 
 \Sigma^{(p)} &=& {\Sigma}^{(p)*} +   \Delta \Sigma^{(p)}\;,
  \label{simplypsplit} 
 \end{eqnarray} 
where 
the first collects  all the terms $\tau$  associated with the elements $\tau \in S_p^{D*}$, i.e. 
 \begin{eqnarray} 
 \Sigma^{(p)*}&:=& \sum_{\sigma \in S_p^D}  \sum_{\tau \in S^{D*}_p}  C_{[\sigma\tau^{-1}]}  \; \Delta\rho{[\sigma]}\;  \Delta H{[\tau]}
\nonumber  \\ &=& \| \Delta \hat{H}\|_2^{p}  \sum_{\sigma \in  S_p^{D}} \sum_{\tau \in S^{D*}_p} C_{[\sigma\tau^{-1}]}  \; 
  \Delta \rho{[\sigma]}\;, \label{sigmatau} 
\end{eqnarray} 
which we simplified invoking  Eq.~(\ref{df12}), and  where the second contribution includes instead  all the remaining derangements $\tau$, i.e. 
\begin{equation}  \label{sigmatau3}
  \Delta \Sigma^{(p)}:=\sum_{\sigma \in S_p^{D}} 
  \sum_{\tau \in S^{D}_p / S^{D*}_p}C_{[\sigma\tau^{-1}]}  \;   \; \Delta \rho{[\sigma]}\; \Delta H{[\tau]} \;.
 \end{equation}

 Observe next  that since all the permutations $\tau$ entering in the sum 
 (\ref{sigmatau}) belong to the same conjugacy class $S_p^{D*}$, we have 
 \begin{eqnarray} 
 C_{[\sigma\tau^{-1}]}  \Delta \rho{[\sigma]} &=& C_{[\sigma (\sigma^{-1}_1 \tau_\star^{-1} \sigma_1)]}  \Delta \rho{[\sigma]} 
\nonumber \\
&=& C_{[(\sigma_1^{-1} \sigma_1) \sigma 
\sigma^{-1}_1 \tau_\star^{-1} \sigma_1]}\Delta  \rho{[\sigma]} 
\nonumber \\
&=& C_{[ \sigma_1 \sigma 
\sigma^{-1}_1 \tau_\star^{-1}]}\Delta  \rho{[\sigma]} 
 = C_{[ \sigma'  \tau_\star^{-1}]} \Delta \rho{[\sigma^{-1}_1\sigma'\sigma_1]} 
\nonumber \\
&=&C_{[ \sigma'  \tau_\star^{-1}]} \Delta \rho{[\sigma']}\;, 
 \end{eqnarray} 
 where in the first identity we invoked (\ref{conjugation})  to write $\tau=\sigma_1 \tau_{\star} \sigma_1^{-1}$  where 
 $\tau_{\star}$ is a fixed element of $S_p^{D*}$ and $\sigma_1\in S_p$, where in the forth line we introduced
 the permutation 
 $\sigma'= \sigma_1 \sigma\sigma_1^{-1}$ which inherits from $\sigma$ the property of being a derangement, and where finally 
 in the third and fifth identity we use the fact $C_{[\sigma]}$ and $\Delta\rho[\sigma]$ are functional of the conjugacy classes. From the above identity we can conclude that $\sum_{\sigma \in S_p^{D}} C_{[\sigma\tau^{-1}]}  \rho{[\sigma]}$
 is independent from the specific choice of $\tau \in S_p^{D*}$, i.e. 
 \begin{equation} 
 \sum_{\sigma \in S_p^{D}} C_{[\sigma\tau^{-1}]} \Delta \rho{[\sigma]} = \sum_{\sigma \in S_p^{D}} C_{[\sigma\tau^{-1}_{\star}]}  \Delta \rho{[\sigma]} \;, \quad \forall \tau \in S_p^{D*}\;.
 \end{equation} 
 Equation~(\ref{sigmatau}) can then be simplified as follows 
  \begin{eqnarray} \Sigma^{(p)*}&=& (p-1)!!  \| \Delta \hat{H}\|_2^{p}  \sum_{\sigma \in S_p^{D}} C_{[\sigma\tau_{\star}^{-1}]}  \; \Delta\rho{[\sigma]} \;. 
\label{summation1} 
\end{eqnarray}
Now we single out from the above summation the term $\sigma=\tau_{\star}$ from the rest
identifying the contributions 
 \begin{eqnarray} 
\overline{\Sigma}^{(p)*} \label{OVER} 
 &:=& (p-1)!!  \| \Delta \hat{H}\|_2^{p}  \;  C_{[1^p]} \;  \; \Delta \rho{[\tau_{\star}]}\;, \\
 \underline{\Sigma}^{(p)*}
 &:=&  (p-1)!!  \| \Delta \hat{H}\|_2^{p}  \sum_{\sigma \in S_p^{D}/\{ \tau_{\star} \} } C_{[\sigma\tau_{\star}^{-1}]}  \; \Delta\rho{[\sigma]} \;, \nonumber \\  \label{UNDER} 
 \end{eqnarray}
where $[1^p]$ is the conjugacy class of the identical permutation. 
Accordingly Eq.~(\ref{summation1}) becomes 
 \begin{eqnarray}  
 \Sigma^{(p)*}&=& \overline{\Sigma}^{(p)*} + \underline{\Sigma}^{(p)*}\;,
 \label{split_sigmaasterisco}
  \end{eqnarray}
which replaced into Eq.~(\ref{simplypsplit}) results into the following decomposition of the $p$-th centred moment of the distribution $P(E; \hat{\rho}, \hat{H})$, 
 \begin{eqnarray} 
 \Sigma^{(p)} &=& \overline{\Sigma}^{(p)*} +   \underline{\Sigma}^{(p)*}  + \Delta \Sigma^{(p)}\;.
  \label{simplypsplit0} 
 \end{eqnarray}

 \subsubsection{Asymptotic behaviour of  the leading term} 
Here we analyze the asymptotic behaviour of $\overline{\Sigma}^{(p)*}$ of Eq.~(\ref{OVER}) which,
as will shall see is the leading term of the centred moment ${\Sigma}^{(p)}$.
Indeed given that  $\tau_{\star}$ is by definition an element of $S_p^{D*}$ from Eq.~(\ref{df12}) we get 
\begin{eqnarray} 
\Delta \rho[\tau_{\star}] = \| \Delta \hat{\rho} \|_2^p\;, \label{deltastar} 
\end{eqnarray} 
while from Eqs.~(\ref{C_1}) and (\ref{sign_of_C}) we get $C_{[1^p]} = \frac{1}{d^p} + \mathcal{O}\left( \frac{1}{d^{p+2}} \right)$
which together leads to 
\begin{eqnarray} 
\overline{\Sigma}^{(p)*} \label{OVERstima} 
 &\simeq& \frac{(p-1)!!  \| \Delta \hat{H}\|_2^{p}  \;   \| \Delta \hat{\rho} \|_2^p }{d^p} \;, 
 \end{eqnarray}
that in the large $d$ limit scales exactly as ${\mathbf G}_{p}$ of Eq.~(\ref{SCALINGNONPURE}) and hence as $\Sigma_G^{(p)}$. 
More precisely, under the hypothesis~(\ref{CONDDD}) we can use  Eq.~(\ref{bounds_su_C_var}) to write 
 \begin{multline}
\left\lvert C_{[1^{p}]} - \tfrac{1}{(d^2-1)^{p/2}} \right\rvert \leq  
\left( 1 - \tfrac{6p^{7/2}}{d^2} \right)^{-1}  \tfrac{6p^{7/2}}{(d^2 - 1)^{p/2} d^{2}}\;,
\label{bound_secondordineC1}
\end{multline}
and hence 
\begin{eqnarray}  \nonumber 
{\left|  \overline{\Sigma}^{(p)*} - \Sigma_G^{(p)} \right|}&=& (p-1)!!  \| \Delta \hat{H}\|_2^{p}  \;   \| \Delta \hat{\rho} \|_2^p 
  \left| C_{[1^p]}  - \tfrac{1}{(d^2-1)^{p/2}}\right| 
\label{boundSigmaOVER} \nonumber \\
& \leq&  {\mathbf G}_{p} \; \overline{f}^*(d,p) \;,
 \end{eqnarray} 
with 
\begin{eqnarray} \label{defoverf} 
\overline{f}^*(d,p)  := \left( 1 - \tfrac{6p^{7/2}}{d^2} \right)^{-1}  \tfrac{6p^{7/2}}{d^{2}}\;.
\end{eqnarray}

 \subsubsection{Asymptotic behaviour of  $\underline{\Sigma}^{(p)}$ (discarded $\sigma$'s)} 
 \label{secerrorsigma} 

Here we evaluate the asymptotic behaviour of  $\underline{\Sigma}^{(p)*}$ of Eq.~(\ref{UNDER}). 
Let's start observing that 
\begin{eqnarray} 
 |\underline{\Sigma}^{(p)*}| 
 &\leq &  (p-1)!!  \| \Delta \hat{H}\|_2^{p}  \sum_{\sigma \in S_p^{D}/\{ \tau_{\star} \} } \left| C_{[\sigma\tau_{\star}^{-1}]}  \;  \Delta\rho{[\sigma]} \right|  \nonumber \\ 
 &\leq &  (p-1)!!  \| \Delta \hat{H}\|_2^{p}   \| \Delta \hat{\rho}\|_2^{p}  \sum_{\sigma \in S_p^{D}/\{ \tau_{\star} \} } \left| C_{[\sigma\tau_{\star}^{-1}]} \right|  \nonumber \\
 &\leq &  (p-1)!!  \| \Delta \hat{H}\|_2^{p}   \| \Delta \hat{\rho}\|_2^{p}  \sum_{\sigma \in S_p/\{ Id \} } \left| C_{[\sigma]} \right|  \;,
  \label{UNDERineq} 
 \end{eqnarray}
where in  the second inequality we used the results of Sec.~\ref{Sec:useful_ineq} to claim that 
\begin{eqnarray} \label{THEBOUND1rho}
|\Delta\rho{[\sigma]}| \leq   \| \Delta \hat{\rho}\|_2^{p}  \qquad \forall \sigma \in S_p^{D}\;,
\end{eqnarray}
while in the third inequality we extended the summation over $\sigma$ outside the derangement set, while still keeping out the terms which makes 
$\sigma\tau_{\star}^{-1}$ the identity permutation $Id$. 
Next step is to provide an estimation of $\sum_{\sigma \in S_p/\{ Id \} } \left| C_{[\sigma]} \right|$.
For this purpose let us decompose $S_p$ in terms of the subsets $S_p(k)$ introduced in Eq.~(\ref{defspk}) and notice that 
\begin{equation}
\left. \begin{array}{l} S_p = \cup_{k=0}^{p} S_p(k)\\\\
\{ Id\} = S_p(0) 
\end{array} 
 \right\} \Longrightarrow S_p/\{ Id \}  = \cup_{k=1}^{p} S_p(k)\;. 
 \end{equation} 
 From this it hence  follows that 
 \begin{eqnarray} 
 \sum_{\sigma \in S_p/\{ Id \} } \left| C_{[\sigma]} \right| &=& \sum_{k=1}^p  \sum_{\sigma \in S_p(k) } \left| C_{[\sigma]} \right| \label{bound_sigmascartate} \\
&\leq&  \left( 1 - \tfrac{6p^{7/2}}{d^2} \right)^{-1} \sum_{k=1}^p \sum_{\sigma \in S_p(k)} 
\tfrac{\mathcal{C}_{[\sigma]}}{d^{p + k}} 
 \nonumber \\
&\leq& \left( 1 - \tfrac{6p^{7/2}}{d^2} \right)^{-1} \sum_{k=1}^p
\tfrac{Cat_p}{d^{p+k}} \left[ {p \atop p-k} \right]
 \nonumber \\
&=& \left( 1 - \tfrac{6p^{7/2}}{d^2} \right)^{-1} \tfrac{Cat_p}{d^{p}}  \left( \tfrac{(p+d-1)!}{d^p (d-1)!} - 1\right) \nonumber \;,
\label{sommasuiCsigma}
\end{eqnarray}
where in the first inequality we adopted Eq.~(\ref{bounds_su_C}) which holds under the condition (\ref{CONDDD}); in the second
 we used~\cite{STANLEY}  
 \begin{equation}
  \max_{\sigma \in S_p(k)} \label{stanley2} 
 \mathcal{C}_{[\sigma]} =  \max_{\sigma \in S_p(k)}\left[ \prod_{i=1}^{p-k} Cat_{{\alpha^{(i)}_{[\sigma]}}-1} \right]\leq Cat_{p}\;,
 \end{equation} 
 and compute the total number of elements in $S_p(k)$ via the String number of the fist kind $\left[ {p \atop p-k} \right]$~\cite{ConcreteMathematics};
and  finally in the last identity we used~(\ref{somma_stirling}). 
Replacing~(\ref{bound_sigmascartate}) and (\ref{SCALINGNONPURE}) into Eq.~(\ref{UNDERineq}) we hence get 
   \begin{eqnarray} 
 \left| \underline{\Sigma}^{(p)*}\right| 
 &\leq&  {\mathbf G}_{p}   \left( 1 - \tfrac{6p^{7/2}}{d^2} \right)^{-1} \;Cat_p\left( \tfrac{\sqrt{d^2-1}}{d}\right)^p \;
\left( \tfrac{(p+d-1)!}{d^p (d-1)!} - 1 \right)  \nonumber\\
&\leq&  {\mathbf G}_{p}   \left( 1 - \tfrac{6p^{7/2}}{d^2} \right)^{-1} \;Cat_p  \;
\left( \tfrac{(p+d-1)!}{d^p (d-1)!} - 1 \right)\;.
 \label{UNDERnewq} 
 \end{eqnarray}
 Observe that, for $p$ constant, in the large $d$ limit the  function multiplying ${\mathbf G}_{p}$ on the r.h.s. tends to zero as $\mathcal{O}(1/d)$. Indeed for 
 \begin{eqnarray} p^2 \leq d, \label{CONDDD1} \end{eqnarray} 
 noticing  that 
\begin{eqnarray}
\ln \frac{(p+d-1)!}{d^p(d-1)!} &=& \ln \prod_{i=0}^{p-1} \left( 1 + \frac{i}{d} \right) 
= \sum_{i=0}^{p-1} \ln  \left( 1 + \frac{i}{d} \right) \nonumber \\ &\leq&  \sum_{i=0}^{p-1} \frac{i}{d} = \frac{p(p-1)}{2d}  \; ,
\end{eqnarray}
and using the inequality
\begin{eqnarray}
e^x - 1 \leq 2x  \quad (x \leq 1.256) \; ,
\end{eqnarray}
with $x = p^2 / d \geq p(p-1)/d$, 
 the rising factorial can be bounded as
\begin{eqnarray}
\frac{(p+d-1)!}{d^p(d-1)!} - 1 \leq \frac{p^2}{d} \; .
\label{bound_risingfact}
\end{eqnarray}
 Exploiting this relation Eq.~(\ref{UNDERnewq})  can finally be casted in the form
    \begin{eqnarray} 
 \left| \underline{\Sigma}^{(p)*}\right| 
\leq {\mathbf G}_{p}  \;   \underline{f}^*(d,p)  \;, 
 \label{UNDERnewq1c} 
 \end{eqnarray}
 with 
     \begin{equation} \label{def2} 
  \underline{f}^*(d,p)  :=   \left( 1 - \tfrac{6p^{7/2}}{d^2} \right)^{-1} \;  \tfrac{p^2 Cat_p}{d}  \;
\;,
 \end{equation}
 the expression being valid when both the conditions (\ref{CONDDD}) and (\ref{CONDDD1}) apply, or in brief when
 Eq. (\ref{CONDONP}) holds.

 \subsubsection{Asymptotic behaviour of  $\Delta{\Sigma}^{(p)}$ (discarded $\tau$'s)} 
 \label{secerrortau} 

Here we study the term $\Delta{\Sigma}^{(p)}$ which according to Eq.~(\ref{sigmatau3}) is obtained by 
running the sum over $\tau$ by
only including derangements that are not in the special subset $S_p^{D*}$.
 Accordingly invoking  Eq.~(\ref{THEBOUND1H}) and (\ref{THEBOUND1rho}) we can write:
 \begin{eqnarray}
 | \Delta \Sigma^{(p)}|&\leq& \sum_{\sigma \in S_p^{D}} 
  \sum_{\tau \in S^{D}_p / S^{D*}_p} | C_{[\sigma\tau^{-1}]}|  \;   \; |\Delta \rho{[\sigma]}|\; |\Delta H{[\tau]}| \nonumber \\
  &\leq&  \eta_{\Delta \hat{H}} \|\Delta \hat{H}\|_2^{p} \| \Delta \hat{\rho}\|_2^{p}  \sum_{\sigma \in S_p^{D}} 
  \sum_{\tau \in S^{D}_p / S^{D*}_p} | C_{[\sigma\tau^{-1}]}|   \nonumber \\
  &\leq&  \eta_{\Delta \hat{H}} \|\Delta \hat{H}\|_2^{p} \| \Delta \hat{\rho}\|_2^{p}  \sum_{\sigma \in S_p} 
  \sum_{\tau \in S^{D}_p / S^{D*}_p} | C_{[\sigma\tau^{-1}]}|   \nonumber \\
  &=&  \eta_{\Delta \hat{H}} \|\Delta \hat{H}\|_2^{p} \| \Delta \hat{\rho}\|_2^{p}  \sum_{\sigma \in S_p} 
  \sum_{\tau \in S^{D}_p / S^{D*}_p} | C_{[\sigma]}|   \nonumber \\
   &=&  \eta_{\Delta \hat{H}} \|\Delta \hat{H}\|_2^{p}   \| \Delta \hat{\rho}\|_2^{p}  \left|S^{D}_p / S^{D*}_p\right| \sum_{\sigma \in S_p} 
  | C_{[\sigma]}|   \nonumber \\
     &\leq&  \eta_{\Delta \hat{H}} \|\Delta \hat{H}\|_2^{p}   \| \Delta \hat{\rho}\|_2^{p} \left( \left\lceil {p!}/{e}  \right\rceil -  (p-1)!!\right)
     \nonumber \\
      && \qquad \qquad  \times    \left( 1 - \tfrac{6p^{7/2}}{d^2} \right)^{-1} \tfrac{Cat_p}{d^{p}}  \tfrac{(p+d-1)!}{d^p (d-1)!}  \;,  \nonumber \\
        \label{sigmatau3ineq}
 \end{eqnarray} 
 where we used~(\ref{stima}) and~(\ref{number_of_derangements}) to get an estimate of $\left|S^{D}_p / S^{D*}_p\right|$  and follow the same passages of 
 Eq.~(\ref{bound_sigmascartate}) to get
\begin{eqnarray} 
 \sum_{\sigma \in S_p } \left| C_{[\sigma]} \right| \leq \left( 1 - \tfrac{6p^{7/2}}{d^2} \right)^{-1} \tfrac{Cat_p}{d^{p}}  \tfrac{(p+d-1)!}{d^p (d-1)!} \;.
\end{eqnarray}
Using hence  Eq.~(\ref{SCALINGNONPURE}) we can then translate Eq.~(\ref{sigmatau3ineq}) into 
 \begin{eqnarray}
&&  | \Delta \Sigma^{(p)}|        \nonumber
      \leq  \eta_{\Delta \hat{H}} {\mathbf G}_{p} \left( \tfrac{\left\lceil {p!}/{e}  \right\rceil }{(p-1)!!} -  1\right) \left( 1 - \tfrac{6p^{7/2}}{d^2} \right)^{-1} \left( \tfrac{\sqrt{d^2-1}}{d}\right)^p 
\\
      && \qquad \qquad \qquad   \times      \;Cat_p\; \nonumber
 \tfrac{(p+d-1)!}{d^p (d-1)!}  \\
  &&\qquad \leq  \eta_{\Delta \hat{H}} {\mathbf G}_{p} \left( \tfrac{\left\lceil {p!}/{e} \right\rceil }{(p-1)!!} -  1\right)
 \left( 1 - \tfrac{6p^{7/2}}{d^2} \right)^{-1}  \;Cat_p \;
 \tfrac{(p+d-1)!}{d^p (d-1)!}  \;,  \nonumber 
 \\ \label{sigmatau3ineqnew}
 \end{eqnarray} 
where, using Eq.~(\ref{bound_risingfact}) and assuming (\ref{CONDDD1}),  the last coefficient  can be further bounded from above as 
\begin{eqnarray}
\frac{(p+d-1)!}{d^p(d-1)!}  \leq 1+ \frac{p^2}{d} \; ,
\label{bound_risingfact1}
\end{eqnarray}
that shows that  in the large $d$ limit this term converges to $1$.
To summarize we can hence conclude that 
\begin{eqnarray}  \label{BOUNDdf} 
  | \Delta \Sigma^{(p)}| &\leq&  \eta_{\Delta \hat{H}} \;    {\mathbf G}_{p} 
\; \Delta f(d,p) \;,
 \end{eqnarray} 
 with 
 \begin{equation} \label{defDELTAf} 
 \Delta f(d,p) := 
\left( \tfrac{\left\lceil {p!}/{e} \right\rceil }{(p-1)!!} -  1\right)
 \left( 1 - \tfrac{6p^{7/2}}{d^2} \right)^{-1}  \;Cat_p ( 1+ \tfrac{p^2}{d}) 
   \;, 
 \end{equation} 
 which again is valid under the condition~(\ref{CONDONP}).

\subsubsection{Total scaling} 
From Eqs.~(\ref{boundSigmaOVER}), (\ref{UNDERnewq1c}), and (\ref{BOUNDdf}) 
 we can now estimate the distance of $\Sigma^{(p)}$ from $\Sigma_G^{(p)}$: 
indeed from Eq.~(\ref{simplypsplit0}), 
we get that under the assumption~(\ref{CONDONP}) we can write 
 \begin{eqnarray} 
 | \Sigma^{(p)} - \Sigma_G^{(p)} | &\leq &| \overline{\Sigma}^{(p)*} -\Sigma_G^{(p)} |  +   |\underline{\Sigma}^{(p)*}| + |\Delta \Sigma^{(p)}| 
 \nonumber \\
 &\leq&  {\mathbf G}_{p}  \;  f_{\hat{H}}(d,p)\;,
  \label{simplypsplit01} 
 \end{eqnarray}  
 with 
 \begin{equation}
  f_{\hat{H}}(d,p) :=  \overline{f}^*(d,p) +  \underline{f}^*(d,p) + \eta_{\Delta\hat{H}} \;  \Delta f(d,p)\;,
 \end{equation}
  exactly matching the expression given in  Eq.~(\ref{sezII_bbound_erroresusigma1}).

\subsection{Case $p$ odd} 
The proof of Eq.~(\ref{sezII_bbound_erroresusigma1})
in the case of $p$ odd  is easier to treat as now  the set  $S_p^{D*}$ is empty, i.e. a condition which we can 
summarize by saying that $\Sigma^{(p)*}=0$. Accordingly we can 
treat the whole $\Sigma^{(p)}$ as we treated the component $\Delta{\Sigma}^{(p)}$ of the even $p$ case.
In particular the entire derivation of Sec.~\ref{secerrortau} can be repeated with the minor modification that 
according to (\ref{THEBOUND1}) the coefficient $\eta_{\hat{H}}$ entering into (\ref{BOUNDdf}) 
 gets replaced by $\sqrt{\eta_{\hat{H}}}$. Therefore in this case we have 
\begin{eqnarray}   | \Sigma^{(p)}| 
&\leq&
\sqrt{\eta_{\hat{H}}}  \; {\mathbf G}_{p}  
\; \Delta f(d,p) \;, \label{sigmatau3ee}
 \end{eqnarray} 
 with $\Delta f(d,p)$ as in Eq.~(\ref{defDELTAf}), 
which  corresponds to the expression given in Eq.~(\ref{intro_tauscartati_simpler}).

\section{Moment generating function}

\label{sec:moments_mgf}

In this section we shall estimate the distance between  the generating function of the 
 moments of the distribution $P(E; \hat{\rho}, \hat{H})$, i.e. 
\begin{equation}
\mathcal{G}(t) := \left\langle e^{t\left(E - \mu \right)} \right\rangle = \sum_{p=0}^\infty \frac{t^p \Sigma^{(p)}}{p!}\;,
\label{moment_genarating_function}
\end{equation}
and the  one associated with the  Gaussian distribution $P_G^{(\mu,\Sigma^{(2)})}(E)$, i.e. 
\begin{equation}
\sum_{p=0}^\infty \frac{t^p \Sigma_G^{(p)}}{p!} =e^{\frac{1}{2}t^2\Sigma^{(2)}}\;,
\end{equation}
where $\Sigma_G^{(p)}$ are the Gaussian moments  defined in Eq.~(\ref{intro_momenti_gaussiani}).

In particular we shall show that, for values of $\lvert t \rvert$ such that
\begin{eqnarray}
\lvert t \rvert \leq \min\left\{ \frac{\sqrt{N^\star}}{\sqrt{\Sigma^{(2)}}},  \frac{{N^\star}}{ {\Delta E}_{\max}
}\right\} \;,
\label{intervallo_validita_t}
\end{eqnarray}
where
\begin{equation} 
{\Delta E}_{\max} := \max\left\{ \left\rvert\mbox{Tr}[\hat{\rho}^{(\downarrow)}  \hat{H}] - \mu\right\rvert, 
\left\rvert\mbox{Tr}[\hat{\rho}^{(\uparrow)}  \hat{H}]  - \mu\right\rvert\right\} \;.
\label{stima_momenti_1.1} 
\end{equation}
and
\begin{eqnarray}
N^\star := \min \left\{ \left\lfloor \sqrt{d} \right\rfloor  , \left\lfloor (d / 2\sqrt{3})^{4/7} \right\rfloor \right\}\;,
\label{def_Nstar}
\end{eqnarray}
 then the following inequality holds:
\begin{eqnarray}
\left\lvert \mathcal{G}(t) - e^{\frac{1}{2}t^2\Sigma^{(2)}} \right\rvert \leq  
e^{2t^2 \Sigma^{(2)}} - 2 t^2  \Sigma^{(2)} - 1 \nonumber \\
+
\tfrac{32t^2 \Sigma^{(2)}  }{\sqrt{\pi} d} 
\left(e^{16t^2 \Sigma^{(2)}} - 1  \right) \;  \nonumber \\
+ \tfrac{1}{2\sqrt{\pi}} \left[ \tfrac{e^{\left[ - 4\lvert t \rvert \sqrt{\Sigma^{(2)} / \eta_{\hat{H}}} \right]}}{ \left( 1 - 4\lvert t \rvert \sqrt{\eta_{\hat{H}}\Sigma^{(2)}} \right)^{\tfrac{1}{\eta_{\hat{H}}}}} 
-  e^{16t^2 \Sigma^{(2)}} \right]
\nonumber \\ +
\tfrac{1}{\left\lceil N^\star/2 \right\rceil !} \tfrac{t^2 \Sigma^{(2)}}{1 - \left( t^2\Sigma^{(2)} / N^\star \right)}
+ 
\tfrac{1}{N^\star!}  \tfrac{\lvert t \rvert {\Delta E}_{\max}}{1 - \left( \lvert t \rvert {\Delta E}_{\max} / N^\star \right)} \nonumber \\
= \frac{32}{3} \sqrt{\eta_{\hat{H}} / \pi} \cdot \left\lvert t \sqrt{ \Sigma^{(2)}} \right\rvert^3  + \mathcal{O} \left( \eta_{\hat{H}} t^4 \left( \sqrt{\Sigma^{(2)}} \right)^2  \right) \;.\nonumber \\
\label{result_mgf}
\end{eqnarray}
The error~(\ref{result_mgf}) is also valid for the characteristic function of the distribution, namely  $\left\langle e^{-itE} \right\rangle$~\cite{NOTE5}.

\subsection{Initial considerations}
We shall estimate the maximum distance between the moment generating functions with the series
\begin{equation}
\left\lvert \mathcal{G}(t) - e^{\frac{1}{2}t^2\Sigma^{(2)}} \right\rvert \leq 
\sum_{p=0}^\infty \frac{t^p}{p!} \left\lvert \Sigma^{(p)} - \Sigma_G^{(p)} \right\rvert.
\label{sommatoria_mgf}
\end{equation}
All the estimations we made in Sec.~\ref{sec:bound_momenti} are valid only under the assumption
~(\ref{CONDONP}).
At some point, the momenta $\Sigma^{(p)}$ will become much more different from $\Sigma^{(p)}_G$, because the distribution $P(E| \hat{\rho}; \hat{H})$ has a compact support, while the gaussian distribution $P_G^{(\mu,\Sigma^{(2)})}(E)$ has infinite tails which become more and more important for the momenta of higher order.
For this reason, it is convenient to consider separately the terms of the sum on the r.h.s. of Eq.~(\ref{sommatoria_mgf}) with 
$p$ not fulfilling~(\ref{CONDONP}) from the rest 
writing 
\begin{equation}
\left\lvert \mathcal{G}(t) - e^{\frac{1}{2}t^2\Sigma^{(2)}} \right\rvert 
\leq  
\Gamma^{*}_{<} + \Delta\Gamma_{<} + \Gamma_{>}\;, 
\end{equation}
where, given the cutoff~(\ref{def_Nstar}),
we set 
\begin{eqnarray}
\Gamma_{>} &:=& \sum_{p = N^\star}^\infty  \frac{t^p}{p!} \left\lvert \Sigma^{(p)} - \Sigma_G^{(p)} \right\rvert \label{GammaMaggiore} \;,
\end{eqnarray}
and, invoking the decomposition~(\ref{simplypsplit}), 
\begin{eqnarray}
\Gamma^{*}_{<} &:=& \sum_{p=1}^{N^\star}  \frac{t^p}{p!} \left\lvert \Sigma^{(p)*} - \Sigma_G^{(p)} \right\rvert = \sum_{\substack{p=4  \\ p \text{ even}}}^{N^\star}  \frac{t^p}{p!} \left\lvert \Sigma^{(p)*} - \Sigma_G^{(p)}  \right\rvert \;,  \nonumber \\ \label{GammaAstMinore}
 \\
\Delta \Gamma_{<} &:=& \sum_{p=1}^{N^\star}  \frac{t^p}{p!} \left\lvert \Delta \Sigma^{(p)} \right\rvert \label{DeltaGammaMinore} \;, 
\end{eqnarray}
with $\Sigma^{(p)*}$ and $\Delta \Sigma^{(p)})$ defined as in
Eqs.~(\ref{sigmatau}) and~(\ref{sigmatau3}) -- notice that Eq.~(\ref{GammaAstMinore}) was simplified by using the fact that 
for $p$ odd both $\Sigma^{(p)*}$ and $\Sigma_G^{(p)}$ are equal to zero, while by definition  $\Sigma^{(2)*}=\Sigma_G^{(2)}$.

\subsection{Bound on $\Gamma_{>}$}
When $p$ is large enough the condition~(\ref{CONDONP}) brakes and 
we can not use anymore the bounds of Sec.~\ref{sec:bound_momenti}. However since the spectrum of the random
variable $E$ is limited as in Eq.~(\ref{sorting_inequality}), for all $p$ we can write 
\begin{eqnarray}\label{stima_momenti_1}
\Sigma^{(p)} \leq {\Delta E}_{\max}^p\;, \end{eqnarray} 
with ${\Delta E}_{\max}^p$ as defined in~(\ref{stima_momenti_1.1}).

Therefore we can write 
\begin{equation}
 \left\lvert \Sigma^{(p)} - \Sigma^{(p)}_G \right\rvert
 =
  \left\lvert \Sigma^{(p)} \right\rvert \leq  {\Delta E}_{\max}^p \;,
\end{equation}
for $p$ odd, and 
\begin{equation}
 \left\lvert \Sigma^{(p)} - \Sigma^{(p)}_G \right\rvert
\leq  \Sigma^{(p)} + \Sigma^{(p)}_G \leq    {\Delta E}_{\max}^p +  \Sigma^{(p)}_G\;, 
\end{equation}
for $p$ even. 
Exploiting these inequalities, when~(\ref{intervallo_validita_t}) holds we can estimate the quantity~(\ref{GammaMaggiore}) as
\begin{multline}
\Gamma_{>} \leq 
\sum_{\substack{p= N^\star \\ p \text{ even}}}^\infty \frac{t^p (p-1)!! \left( \Sigma^{(2)} \right)^{p/2}}{p!} 
+ \sum_{\substack{p = N^\star}}^\infty \frac{t^p  {\Delta E}_{\max}^p }{p!} \\
\leq
\sum_{n=\left\lceil N^\star/2 \right\rceil}^\infty \frac{1}{n!} \left( \frac{t^2 \Sigma^{(2)}}{2} \right)^n
+
\sum_{\substack{p = N^\star}}^\infty \frac{t^p  {\Delta E}_{\max}^p }{p!} \\
\leq
\tfrac{\left( N^\star \right)^{\left\lceil N^\star/2 \right\rceil}}{\left\lceil N^\star/2 \right\rceil !} 
\sum_{n=\left\lceil N^\star/2 \right\rceil}^\infty \left( \frac{t^2 \Sigma^{(2)}}{N^\star} \right)^n \\
+ \frac{
 \left( N^\star \right)^{N^\star}}{ N^\star !}
 \sum_{\substack{p = N^\star}}^\infty \left(\frac{t {\Delta E}_{\max} }{N^{\star}}\right)^p\\
\leq  
\frac{1}{\left\lceil N^\star/2 \right\rceil !} \frac{t^2 \Sigma^{(2)}}{1 - \left( t^2\Sigma^{(2)} / N^\star \right)} \\
+ 
\frac{1}{N^\star!}  \frac{t {\Delta E}_{\max}}{1 - \left( t {\Delta E}_{\max} / N^\star \right)}
\\
= \mathcal{O} \left( \frac{1}{\left\lceil N^\star/2 \right\rceil !} \right)  = \mathcal{O} \left( \frac{1}{\left\lceil \sqrt{d}/2 \right\rceil !} \right)\;. 
\end{multline}

\subsection{Bound on $\Gamma^{*}_{<}$}
From Eq.~(\ref{split_sigmaasterisco}) we can write 
\begin{eqnarray}
\Gamma^{*}_{<} \leq  \overline{\Gamma}^{*}_{<} + \underline{\Gamma}^{*}_{<}\;, 
\end{eqnarray}
where
\begin{eqnarray}
\overline{\Gamma}^{*}_{<} &:=& \sum_{\substack{p=4  \\ p \text{ even}}}^{N^\star}   \frac{t^p}{p!} \left| \overline{\Sigma}^{(p)*} - \Sigma_G^{(p)} \right| \label{GammaAstMinoreOver} \;, \\
\underline{\Gamma}_{<}^* &:=& \sum_{\substack{p=4  \\ p \text{ even}}}^{N^\star}  \frac{t^p}{p!} | \underline{\Sigma}^{(p)*}| \label{GammaAstMinoreUnder} \;,
\end{eqnarray}
and with $\overline{\Sigma}^{(p)*}$ and $\underline{\Sigma}^{(p)*}$ as defined in~(\ref{OVER}) and~(\ref{UNDER}).

First we provide a bound for the term $\overline{\Gamma}^{*}_{<}$. Retrieving from~(\ref{boundSigmaOVER}) the bound on $\left|  \overline{\Sigma}^{(p)*} - \Sigma_G^{(p)} \right|$, and using the expression~(\ref{varianzaCOMP}) for the variance $\Sigma^{(2)}$, we have
\begin{eqnarray}
 \overline{\Gamma}^{*}_{<}  &\leq& \sum_{\substack{p=4 \\ p \text{ even}}}^{\infty}  \tfrac{ (p-1)!!}{p!} 
 \left( 1 - \tfrac{6p^{7/2}}{d^2} \right)^{-1}  \tfrac{6p^{7/2}t^p}{d^{2}} \left( \Sigma^{(2)} \right)^{p/2} \nonumber\\
&\leq& \sum_{\substack{p=4 \\ p \text{ even}}}^{\infty}  \frac{ (p-1)!!}{p!}  
t^p
\left( \Sigma^{(2)} \right)^{p/2} \leq \sum_{n=2}^{\infty}    \frac{2^n t^{2n}}{n!} 
 \left( \Sigma^{(2)} \right)^{n} \nonumber \\
  &=& \exp\left[ 2 t^2   \Sigma^{(2)} \right] - 2t^2 \Sigma^{(2)} - 1 \nonumber \\
&=& \mathcal{O} \left(  t^4 \left( \Sigma^{(2)} \right)^{2} \right) \; ,
\end{eqnarray}
in the second inequality we used the fact that 
the condition $p \leq N^\star \leq \left( d / 2\sqrt{3} \right)^{4/7}$ implies 
\begin{equation}
\tfrac{6p^{7/2}}{d^2} \left( 1 - \tfrac{6p^{7/2}}{d^2} \right)^{-1} \leq 1 \; .
\label{bound_sulmatsu} 
\end{equation}

Regarding $\underline{\Gamma}^{*}_{<}$, recalling equation~(\ref{UNDERnewq1c}), and expliciting the definitions~(\ref{def2}) and~(\ref{SCALINGNONPURE}) of the quantities $\underline{f}^*(d,p)$ and $ {\mathbf G}_{p}$ we have that
\begin{eqnarray}
\left| \underline{\Sigma}^{(p)*}\right| 
\leq
\tfrac{(p-1)!!}{(d^2-1)^{p/2}}  \| \Delta \hat{H}\|_2^p \;  \| \Delta \hat{\rho}\|_2^p
\left( 1 - \tfrac{6p^{7/2}}{d^2} \right)^{-1} \;  \tfrac{p^2 Cat_p}{d} \; , \nonumber \\
= (p-1)!! \left( 1 - \tfrac{6p^{7/2}}{d^2} \right)^{-1} \left( \Sigma^{(2)} \right)^{p/2} \nonumber \\
\end{eqnarray}
from which, using the bound~(\ref{bound_catalan}) on Catalan numbers we can write 
\begin{eqnarray}
\underline{\Gamma}_{<}^*   &\leq&  \sum_{\substack{p=4  \\ p \text{ even}}}^{N^\star}  \tfrac{(p-1)!!}{ p!} \left( 1 - \tfrac{6p^{7/2}}{d^2} \right)^{-1} \;  \tfrac{(4t)^p\sqrt{p} }{\sqrt{\pi}d} \left( \Sigma^{(2)} \right)^{p/2}
\nonumber \\
&\leq&
  2 \sum_{n=2}^{d} \frac{ \left( 2n-1 \right)!!}{(2n)!}  \left( 16t^2 \Sigma^{(2)} \right)^{n} \frac{\sqrt{2n}}{\sqrt{\pi}d} \nonumber \\
&\leq& 
  2 \sum_{n=2}^{\infty} \frac{ \left( 2n-1 \right)!!}{(2n)!}  \left( 16t^2 \Sigma^{(2)} \right)^{n} \frac{{n}}{\sqrt{\pi}d} \nonumber \\
  &=& \frac{1}{\sqrt{\pi} d} \sum_{n=2}^{\infty}\frac{2n}{n!} \left( 16t^2 \Sigma^{(2)} \right)^{n} \nonumber \\
   &=& \frac{1}{\sqrt{\pi} d} \sum_{n=2}^{\infty} \frac{t}{n!} \frac{\partial}{\partial t} \left( 16t^2 \Sigma^{(2)} \right)^{n} \nonumber \\
   &=&\frac{t}{\sqrt{\pi} d} \frac{\partial}{\partial t} \sum_{n=2}^{\infty} \frac{1}{n!}  \left( 16t^2 \Sigma^{(2)} \right)^{n} \nonumber \\
  &=& \frac{t}{\sqrt{\pi} d} \frac{\partial}{\partial t} \left[ \exp\left( 16t^2 \Sigma^{(2)}  \right) - 16t^2 \Sigma^{(2)} - 1  \right] \nonumber \\
  &=& 
  \frac{32t^2  \Sigma^{(2)}  }{\sqrt{\pi} d} 
  \left[ \exp\left( 16t^2 \Sigma^{(2)}  \right) - 1  \right] \; , \nonumber \\
    &=& \mathcal{O} \left( t^4 \left( \Sigma^{(2)} \right)^2 d^{-1} \right) \; , \nonumber \\
\label{expsum_sigma}
\end{eqnarray}
where in the second inequality we let $n = p/2$ and used the fact that, since $p \leq N^\star \leq \left( d / 2\sqrt{3} \right)^{4/7}$,
\begin{equation}
\left( 1 - \tfrac{6p^{7/2}}{d^2} \right)^{-1} \leq 2 \;,
\label{bound_sulmatsu2} 
\end{equation}
and in the third inequality we used $\sqrt{2n} \leq n$.

\subsection{Bound on $\Delta \Gamma_{<}$}
From the definitions~(\ref{DeltaGammaMinore}) and~(\ref{sigmatau3}), we can write the term $\Delta \Gamma_{<}$ as
\begin{eqnarray}
\Delta\Gamma_{<} = 
\sum_{p=1}^\infty \frac{t^p}{p!} \sum_{\sigma \in S_p^{D}} 
\sum_{\tau \in S^{D}_p / S^{D*}_p}C_{[\sigma\tau^{-1}]}  \;   \; \Delta \rho{[\sigma]}\; \Delta H{[\tau]} \;.\nonumber \\
\end{eqnarray}
First we use~(\ref{sommasuiCsigma}), as well as $\Delta\rho{[\sigma\tau]} \leq \| \Delta \hat\rho \|^p_2 $, to write
\begin{eqnarray}
\left \lvert \Delta\Gamma_{<} \right \rvert \nonumber 
\leq
\sum_{p=3}^\infty \frac{t^p}{p!} \| \Delta \rho \|^p_2 \left( \sum_{\sigma \in S_p^{D}} \left\lvert C_{[\sigma]} \right\rvert \right) 
 \left( \sum_{\tau \in S^{D}_p / S^{D*}_p} \; \Delta H{[\tau]} \right) \nonumber \\
< \sum_{p=3}^\infty \frac{t^p \| \Delta \rho \|^p_2 }{p!} \left( 1 - \tfrac{6p^{7/2}}{d^2} \right)^{-1}
\tfrac{Cat_p(d+p-1)!}{d^{2p}(d-1)!} \sum_{\tau \in S^{D}_p / S^{D*}_p} H{[\tau]} \nonumber \\
=
\sum_{p=3}^\infty \frac{t^p}{p!} \left( 1 - \tfrac{6p^{7/2}}{d^2} \right)^{-1}
\tfrac{Cat_p(d+p-1)!}{d^{2p}(d-1)!} \sum_{k=\lfloor p/2 \rfloor + 1}^{p} \sum_{\tau \in S^{D}_p(k)} H{[\tau]}\;, \nonumber \\
\label{Somma_DeltaGamma1}
\end{eqnarray}
where in thethird line we introduced the sets 
\begin{eqnarray}
S^{D}_p(k) := S^{D}_p \cap S_p(k) \; ,
\end{eqnarray}
 and noticed that $S^{D}_p(k) = \varnothing$ when $k < \lfloor p/2 \rfloor + 1$ and that, when $p$ is even, $S^{D}_p(p/2) = S^{D*}_p$ (recall also that $S^{D*}_p = \varnothing$ in the case of odd $p$).

Replacing the bounds~(\ref{bound_catalan}) and~(\ref{bound_risingfact}) in~(\ref{Somma_DeltaGamma1}), we have
\begin{eqnarray}
\left \lvert \Delta\Gamma_{<} \right \rvert \nonumber
\leq
\sum_{p=3}^\infty \frac{4^p t^p \| \Delta \hat\rho \|^p_2 }{\sqrt{\pi}d^p p^{3/2} p!}  \left( 1 - \tfrac{6p^{7/2}}{d^2} \right)^{-1} \\ \times
\left(1 + \frac{p^2}{d} \right) \sum_{k=\lfloor p/2 \rfloor + 1}^{p} \sum_{\tau \in S^{D}_p(k)} \eta_{\hat{H}}^{k - p/2} \| \Delta \hat{H} \|^p_2 \;.
\label{Somma_DeltaGamma2}
\end{eqnarray}

Now, in order to simplify the calculations, we will use the inequality~(\ref{bound_sulmatsu2}) to infer that, since $4 \leq p \leq \sqrt{d}$,
\begin{equation}
p^{-3/2} \left( 1 - \frac{6p^{7/2}}{d^2} \right)^{-1} \left(1 + \frac{p^2}{d} \right) < 4p^{-3/2} \leq \frac{1}{2}  \; .
\label{pulisci_roba}
\end{equation}
With the help of~(\ref{pulisci_roba}), and noticing from~(\ref{varianzaCOMP}) that $\| \Delta \hat{H} \|_2 \| \Delta \hat{\rho} \|_2 / d < \sqrt{\Sigma^{(2)}}$, we can turn~(\ref{Somma_DeltaGamma2}) into
\begin{eqnarray}
\left \lvert \Delta\Gamma_{<} \right \rvert
\leq
\frac{1}{2\sqrt{\pi}} \sum_{p=3}^\infty \frac{1}{p!}\left( \tfrac{4 t \| \Delta \hat{H} \|_2 \| \Delta \hat{\rho} \|_2 }{d} \right)^p \nonumber \\
\times \sum_{k=\lfloor p/2 \rfloor + 1}^{p} \sum_{\tau \in S^{D}_p(k)} \eta_{\hat{H}}^{k - p/2} \nonumber \\
\leq 
\frac{1}{2\sqrt{\pi}} \sum_{p=3}^\infty \frac{1}{p!}\left( 4 t \sqrt{\Sigma^{(2)}} \right)^{p} \sum_{k=\lfloor p/2 \rfloor + 1}^{p} \sum_{\tau \in S^{D}_p(k)} \eta_{\hat{H}}^{k - p/2} \nonumber \\
= \frac{1}{2\sqrt{\pi}} \sum_{p=3}^\infty \frac{1}{p!}\left( 4t \sqrt{\eta_{\hat{H}} \Sigma^{(2)}} \right)^{p} \sum_{k=\lfloor p/2 \rfloor + 1}^{p} \mathrm{d}_{p, k} \eta_{\hat{H}}^{k-p} \; , \nonumber \\
\label{Somma_DeltaGamma3}
\end{eqnarray}
where we called $\mathrm{d}_{p, n}$ the number of elements in the set $S^{D}_p(k)$.

We do not need the explicit values of $\mathrm{d}_{p, k}$ . They are implicitly given by the coefficients in the expansion of the exponential generating function\cite{Flajolet}

\begin{equation}
\sum_{i=1}^{\infty} \sum_{p=i }^{\infty}  \mathrm{d}_{p, p-i} c^i \frac{x^p}{p!} = \frac{e^{-cx}}{(1-x)^c} \; .
\label{expgensum_derangements}
\end{equation}

In wiev of~(\ref{Somma_DeltaGamma3}), we want to keep in the sum only the derangements with less than $p/2$ cycles ($\tau \not \in [2^{p/2}]$), whose generating functions is
\begin{multline}
\sum_{p=2}^{\infty} \sum_{k=\lfloor p/2 \rfloor - 1}^{p}  \mathrm{d}_{p, k} c^{p-k} \frac{x^p}{p!}
=
 \sum_{p=2}^{\infty} \sum_{i=1}^{\lceil p/2 \rceil - 1}   \mathrm{d}_{p, p-i} c^i \tfrac{x^p}{p!} 
 \\
= \sum_{p=0}^{\infty} \sum_{i=1}^{\lfloor p/2 \rfloor}   \mathrm{d}_{p, p-i} c^i \frac{x^p}{p!}  - \sum_{\substack{p\text{ even}}}^{\infty} \tfrac{(p-1)!!}{p!} c^{p/2} x^{p} \\
=
\frac{e^{-cx}}{(1-x)^c} - e^{\frac{cx^2}{2}}
= 
\frac{cx^3}{3} + \mathcal{O} \left( cx^4 \right) \; .
\label{expgensum_derangementstolti}
\end{multline}

Using~(\ref{expgensum_derangementstolti}) into~(\ref{Somma_DeltaGamma3}), with the identifications $x=4t\sqrt{\eta_{\hat{H}}\Sigma^{(2)}}$ and $c = \eta^{-1}_{\hat{H}}$, we can conclude that
\begin{eqnarray}
\left \lvert \Delta\Gamma_{<} \right \rvert \leq
\frac{1}{2\sqrt{\pi}} \left[ \frac{e^{\left[ - 4\lvert t \rvert \sqrt{\Sigma^{(2)} / \eta_{\hat{H}}} \right]}}{ \left( 1 - 4\lvert t \rvert \sqrt{\eta_{\hat{H}}\Sigma^{(2)}} \right)^{\tfrac{1}{\eta_{\hat{H}}}}} 
-  e^{16t^2 \Sigma^{(2)}} \right]
\nonumber \\ =
\frac{32}{3} \sqrt{\eta_{\hat{H}} / \pi} \cdot \left\lvert t \sqrt{ \Sigma^{(2)}} \right\rvert^3  + \mathcal{O} \left( \eta_{\hat{H}} t^4 \left( \sqrt{\Sigma^{(2)}} \right)^2  \right)\;.
\nonumber \\
\label{Somma_DeltaGamma4}
\end{eqnarray}

\section{Conclusions} \label{Sec:conc} 

We derived a bound for the distance between the energy distribution in the unitary orbit of a quantum state and the normal distribution with the same variance, showing that for large dimensions of the Hilbert space the difference in the charachteristic function is suppressed as $d^{-3}$.
We have also characterised the individual moments of the distribution: these result apply also to unitary $t$-design, which mimic the uniform distribution of unitary transformations up to the $t$-th moments.  
Our findings, therefore, are suitable to be be employed for certifying that an unitary $t$-design behaves as expected.
Since the results do not depend on the specific form of the Hamiltonian (but only on a very general hypotesis on its eigenvalues), they can be applied to every observable that can be measured on the system.

	The validity of our results is limited to the specific (but relevant in the field of quantum computing) case in which the evolution of the system is described by a random unitary matrices drawn from the Haar-uniform ensemble (also known as \emph{Circular Unitary Ensemble} \cite{Dyson1962, Eynard2015}).
	In our opinion, it would be interesting to perform a similar analysis on the work distribution in the case in which the time evolution of the system is given by an Hamiltonian drawn from the Wigner-Dyson Gaussian Unitary Ensemble \cite{Wigner1955, Fyodorov2004}, which describes quantum chaotic systems lacking time-reversal symmetry \cite{Gharibyan2018, mehta2004random}, or from the Gaussian Orthogonal or the Gaussian Symplectic Ensembles, which model the evolution quantum chaotic systems in presence time-reversal symmetry \cite{mehta2004random}. The combinatorical tools necessary for this analogous task have been in part already developed \cite{Matsumoto2013}, but significant more work may be required.

Another possible future developement of this work is generalizing it to Hilbert space of infinite dimension. An immediate difficulty is that the Haar measure on the unitary group of infinite dimension is not well-defined \cite{Oxtoby1946}. However, a class of states of particular interest are the gaussian states on systems charachterised by a $N$-modes quadratic bosonic Hamiltonian\cite{WANG2007}: in this setting, it could be of some interest to study the distribution of their energy over the group of symplectic transformations \cite{serafini2017quantum}.

The Authors acknowledge 
 support from 
PRIN 2017 (Progetto di Ricerca di Interesse Nazionale): project 
“Taming complexity with quantum strategies” QUSHIP (2017SRNBRK).


\begin{thebibliography}{70}
\bibitem[NOT({\natexlab{c}})]{NOTA1}
		{\natexlab{c}}.
		\newblock Notice that the requirement $c[\sigma]\leq d$ is trivial as long as
		$p\leq d$; if we extend the sum in~(\ref{Weingarten}) to all the irreducible
		representations of $S_p$, we obtain the Weingarten function for the special
		unitary group $S\mathbf{U}(d)$. Therefore, when $c[\sigma]<d$; is equivalent
		to do the integral~(\ref{integrale_U}) on the unitary group $\mathbf{U}(d)$
		or on the special unitary group $S\mathbf{U}(d)$.
\end{thebibliography}

\end{document}